\documentclass[11pt,letterpaper]{article}

\usepackage{amsmath,amsthm,amsfonts,color,hyperref,anysize,amssymb,enumerate,sectsty,color,dsfont,graphicx,epstopdf}
\usepackage[section]{placeins}
\usepackage[usenames,dvipsnames]{xcolor}
\usepackage[normalem]{ulem}
\usepackage[margin=1.0in]{geometry}

\sectionfont{\fontsize{14}{14}\selectfont}

\usepackage[english]{babel} % English language/hyphenation
\usepackage{dsfont}

\usepackage{lipsum} % Used for inserting dummy 'Lorem ipsum' text into the template

\usepackage{sectsty} % Allows customizing section commands

\usepackage{amsthm}
\newtheorem{theorem}{Theorem}[section]
\newtheorem{lemma}{Lemma}[section]
\newtheorem{definition}{Definition}[section]
\newtheorem{observation}{Observation}[section]
\newtheorem{assumption}{Assumption}[section]
\newtheorem{corollary}{Corollary}[section]
\newtheorem{fact}{Fact}[section]
\numberwithin{equation}{section} % Number equations within sections (i.e. 1.1, 1.2, 2.1, 2.2 instead of 1, 2, 3, 4)
\numberwithin{figure}{section} % Number figures within sections (i.e. 1.1, 1.2, 2.1, 2.2 instead of 1, 2, 3, 4)
\numberwithin{table}{section} % Number tables within sections (i.e. 1.1, 1.2, 2.1, 2.2 instead of 1, 2, 3, 4)

\newcommand\numberthis{\addtocounter{equation}{1}\tag{\theequation}}

\def\bdi{\mathbf{d}_i}
\def\bd{\mathbf{d}}
\def\bk{\mathbf{k}}
\def\bx{\mathbf{x}}
\def\by{\mathbf{y}}
\def\bxi{\mathbf{x}_i}
\def\bn{\mathbf{n}}
\def\bp{\mathbf{p}}
\def\bq{\mathbf{q}}

\def\bb{\mathbf{b}}
\def\bv{\mathbf{v}}

\def\bi{b_i}
\def\bmi{b_{-i}}

\def\pj{p_j}
\def\qj{q_j}
\def\rj{r_j}

\def\pjs{p_j^*}

\def\ei{e_i}

\def\ui{u_i}
\def\vi{v_i}

\def\vmi{v_{-i}}

\def\exi{\mathbf{x}_i}
\def\del{\delta}

\def\sw{\text{\rm SW}}
\def\ne{\text{\rm NE}}
\def\opt{\text{\rm OPT}}

\def\db{\,\db}
\def\dv{\,\dv}

\def\eps{\epsilon}

\def\PoA{\text{PoA}}
\def\GPoA{\text{GPoA}}

\newcommand{\hide}[1]{}

\def\xj{x_{\{j\}}}
\def\zj{z_{\{j\}}}

\title{Large Market Games with Near Optimal Efficiency}

\author{Richard Cole and Yixin Tao\\Courant Institute, NYU}

\date{\normalsize\today} % Today's date or a custom date
\newenvironment{pfof}[1]{\begin{proof}[\emph{\textbf{Proof of #1: }}]}{\end{proof}}

\begin{document}
\maketitle
%\begin{bottomstuff}
%Author's addresses: Richard Cole and Yixin Tao: Computer Science %Department, Courant Institute, New York University.
%\end{bottomstuff}

\maketitle

\begin{abstract}
As is well known, many classes of markets have efficient equilibria,
but this depends on agents being non-strategic, i.e.\ that they declare
their true demands when offered goods at particular prices, or in other
words, that they are price-takers.
An important question is how much the equilibria degrade in the face of
strategic behavior, i.e.\ what is the Price of Anarchy (PoA) of the market viewed
as a mechanism?

\hide{
Recent work by Babaioff et al.~\cite{BabaioffLNP14} showed that for Walrasian Mechanisms,
for suitable classes of markets, the PoA is bounded by small constant factors,
and, by means of examples, they showed
these bounds were unavoidably larger than 1 (e.g.\ $\tfrac 43$ and 2).
}

Often, PoA bounds are modest constants such as $\tfrac 43$ or 2.
Nonetheless, in practice a guarantee that no more than 25\% or 50\% of the economic value
is lost may be unappealing.
This paper asks whether significantly better bounds are possible under plausible
assumptions.
In particular, we look at how these worst case guarantees improve in the following large settings.
\begin{itemize}
\item
Large Walrasian auctions:
These are auctions with many copies of each item and many agents.
We show that the PoA tends to 1 as the market size increases,
under suitable conditions, mainly
that there is some uncertainty about the numbers of copies of each good
and demands obey the gross substitutes condition.
We also note that some such assumption is unavoidable.
\item
Large Fisher markets:
Fisher markets are a class of economies that has received considerable attention in
the computer science literature. A large market is one in which at equilibrium,
each buyer makes only a small fraction of the total purchases;
the smaller the fraction, the larger the market.
Here the main condition is that demands are based
on homogeneous monotone utility functions that satisfy the gross substitutes condition.
Again, the PoA tends to 1 as the market size increases.
\end{itemize}
Furthermore, in each setting, we quantify the tradeoff between market size and the PoA.
\end{abstract} 

\section{Introduction}
\label{sec:intro}

When is there no gain to participants in a game from strategizing?
One answer applies
when players in a game have no prior knowledge;
then a game that is strategy proof ensures 
that truthful actions are a best choice for each player.
However, in many settings there is no strategy proof mechanism.
Also, even if there is a strategy proof mechanism,
with knowledge in hand, other equilibria are possible,
for example, the ``bullying'' Nash Equilibrium as illustrated by the
following example:
there is one item for sale using a second price auction,
the low-value bidder bids an amount at least equal to the value of
the high-value bidder, who bids zero; the resulting equilibrium achieves
arbitrarily small social welfare compared to the optimal outcome.

To make the notion of gain meaningful one needs to specify what the game or mechanism is
seeking to optimize. Social welfare and revenue are common targets.
For the above example, the social welfare achieved in the bullying equilibria
can be arbitrarily far from the optimum.
However, for many classes of games,
over the past fifteen years,
bounds on the gains from strategizing, a.k.a.\ the Price of Anarchy (PoA),
have been obtained, with much progress coming thanks to
the invention of the smoothness methodology
\cite{Roughgarden2015,Roughgarden2012,Syrgkanis:2013:CEM:2488608.2488635,FeldmanILRS:15};
many of the resulting bounds have been shown to be tight.
Often these bounds are modest constants, such as $\tfrac 43$~\cite{RoughgardenT2002}
or $2$~\cite{syrgkanis2012bayesian}, etc., but rarely
are there provably no losses from strategizing, i.e.\ a PoA of 1.

This paper investigates when bounds close to 1 might be possible.
In particular, we study both large Walrasian auctions and large Fisher markets viewed as mechanisms.

\paragraph{Walrasian Auctions} 
Walrasian Auctions are used in settings where there
are goods for sale and agents, called bidders, who want to buy these goods.
Each agent has varying preferences for different subsets of the goods,
preferences that are represented by valuation functions.
The goal of the auction is to identify equilibrium prices;
these are prices at which all the goods sell, and each bidder receives
a favorite (utility maximizing) collection of goods, where each bidder's
utility is quasi-linear: the difference of its valuation for the goods and their cost at the
given prices.
Such prices, along with an associated allocation of goods, are said to form a Walrasian
equilibrium.

Walrasian equilibria for indivisible goods are known to exist when each bidder's demand satisfies the gross substitutes
property~\cite{GulStac99}, but this is the only substantial class of settings in
which they are known to exist.
\hide{
In these settings, the equilibria form a lattice, with the prices at the top of the
lattice being those that result from a Dutch auction, while those at the bottom result from
an English auction.
}

\cite{BabaioffLNP14} analyzed the PoA of the games induced by Walrasian mechanisms,
i.e.\ the  prices were computed by a method, such as an English or Dutch auction,
that yields equilibrium prices when these exist.
Note that the mechanism can be applied even when Walrasian equilibria do not exist, though
the resulting outcome will not be a Walrasian equilibrium.
But even when Walrasian equilibria exist, because bidders may strategize, in general the outcome
will be a Nash equilibrium rather than a Walrasian one.
Among other results, Babaioff et al.~showed an upper bound of 4 on the PoA for any Walrasian mechanism
when the bids and valuations satisfied the gross substitutes property and overbidding was not allowed.\footnote{They also proved a version of this bound which was parameterized w.r.t. the amount of overbidding.}
In addition, they obtained lower bounds on the PoA that were
greater than 1, even when overbidding was not allowed, which
excludes bullying equilibria; e.g.\ the English auction has a PoA of at least 2.

Babaioff et al.\ also noted that the prices computed by double auctions, widely used in financial
settings, are essentially computing a price that clears the market and maximizes trade;
one example they mention is the computation of the opening prices on the New York Stock Exchange,
and another is the adjustment of prices of copper and gold in the London market.

By a large auction, we intend an auction in which there are many copies of each good, and in addition
the demand set of each bidder is small.
The intuition is that then each bidder will have a small influence and hence strategic
behavior will have only a small effect on outcomes.
In fact, this need not be so.
For example, the bullying equilibrium persists: it suffices to increase the numbers
of items and bidders for each type to $n$, and have the buyers of each type
follow the same strategy as before.

What allows this bullying behavior to be effective is the precise match between the number of items
and the number of low-value bidders.
The need for this exact match also arises in the lower-bound examples in~\cite{BabaioffLNP14}
(as with the bullying equilibrium, it suffices to pump up the examples by a factor of $n$).
To remove these equilibria that demonstrate PoA values larger than 1,
it suffices  to introduce some uncertainty regarding the numbers of items and/or bidders.
Indeed, in a large setting it would seem unlikely that such numbers would be known precisely.
We will create this uncertainty by using distributions to determine the number of copies of each good.
This technique originates with~\cite{Swinkels2001}.
In contrast, prior work on non-large markets eliminated the potentially unbounded PoA
of the bullying equilibrium by assuming  bounds on the possible overbidding~\cite{BhawalkarR2011,ChristodoulouKS2008,FuKL2012,Syrgkanis:2013:CEM:2488608.2488635}.

Our main result on large Walrasian auctions is that the PoA of the Walrasian mechanism tends to 1 
as the market size grows.
This result assumes that expected valuations are bounded regardless of the size of the market.
We specify this more precisely when we state our results
in Section~\ref{sec:results}.
This bound applies to both Nash and Bayes-Nash equilibria;
as it is proved by means of a smoothness argument,
it extends to mixed Nash and coarse correlated equilibria,
and outcomes of no-regret learning.

\paragraph{Fisher Markets}
A Fisher market is a special case of an exchange economy in which the agents are either buyers or sellers.
Each buyer is endowed with money but has utility only for non-money goods; 
each seller is endowed with non-money goods,
WLOG with a single distinct good, and has utility only for money.
Fisher markets capture settings in which buyers want to spend all their money.
In particular, they generalize the competitive equilibrium from equal incomes (CEEI)~\cite{HyllandZec1979,Varian1974},
in that they allow buyers to have non-equal incomes.
While at first sight this might appear rather limiting, we note that
much real-world budgeting in large organizations treats budgets as money to be spent in full,
with the consequence that unspent money often has no utility to those making the spending decisions.
The budgets in GoogleAds and other online platforms can also be viewed as money that is intended to be
spent in full.

We consider the outcomes when buyers bid strategically in terms of how they declare their utility functions.
We show that the PoA tends to 1 as the setting size increases.
The only assumptions we need are some limitations on the buyers' utility functions:
they need to satisfy the gross substitutes property and to be monotone and homogeneous of degree 1.

This result is also obtained via a smoothness-type bound and hence extends to bidders
playing no-regret strategies, assuming that the ensuing prices are always bounded away from zero.
We ensure this by imposing reserve prices, but for lack of space this result is deferred to {the full version of the paper}.

\paragraph{Roadmap}
In Section~\ref{sec:prelim} we provide the necessary definitions and background,
in Section~\ref{sec:results} we state our results, 
which are then shown in Sections~\ref{sec:wal-equil} and~\ref{sec:large-fisher},
covering large auctions and large Fisher markets respectively.

\subsection{Related Work}

The results on auctions generalize earlier work of~\cite{Swinkels2001} who showed analogous
results for auctions of multiple copies of a single good. In contrast, we consider auctions
in which there are multiple goods.
Swinkels analyzed discriminatory and non-discriminatory mechanisms.
For the latter, he showed that any mechanism that used a combination of the $k$-th and $(k+1)$-st
prices when there were $k$ copies of the good on sale achieved a PoA that tended to 1 with the
auction size\footnote{Swinkels did not use the then recently formulated PoA terminology to state his result.}.
Our result also weakens some of the assumptions in Swinkels work.

The second closely related work on auctions is due to~\cite{FeldmanILRS:15}.
They also {analyze several} large settings. 
Among other results, they analyze auctions in which the PoA tends to $1$ as their size grows to $\infty$.
Their results are derived from a new type of smoothness argument.
Depending on the result, they require either uncertainty in the number of goods
or the number of bidders.
In contrast, our main result uses a previously known smoothness technique
plus uncertainty in the number of goods.
We contrast the techniques in more detail after we present our result in Section~\ref{sec:wal-equil}.
They also show that for traffic routing problems, the PoA of the atomic case tends to that
of the non-atomic case as the number of units of traffic grows to $\infty$.

The idea of uncertainty in the number of agents or items first arose in the Economics
literature. \cite{Myerson2000} used it in the context of voting games, and~\cite{Swinkels2001}
in the context of auctions.
Later, uncertainty in the number of agents was used with the Strategy Proof in the Large
concept~\cite{AzevBudish2012}.

\cite{hsu2015prices} considered the effects of non-unique demands on the social welfare,
assuming allocations were based on demands. Given a genericity assumption, they showed
that in markets with buyers having matroid based valuations 
the inefficiency was proportional to the number of distinct goods,
and so if this was a constant, the efficiency would tend to 1 as the market size grows.

% As already noted, Babaioff et al.\ gave bounds on the PoA of Walrasian equilibria.

The study of the behavior of large exchange economies was first considered
by~\cite{RobertsPosle1976}, which they modeled as a \emph{replica economy},
the $n$-fold duplication of a base economy, showing that individual utility gains from strategizing
tend to zero as the economy grows.
Subsequently, \cite{JacksonMan97} showed that with some \emph{regularity} assumptions,
the equilibrium allocations converge to the competitive equilibrium.
In contrast, our result proves bounds in terms of a parameter characterizing the size
of the economy.
More recently, \cite{AlNajjar:2007} studied the efficiency of
exchange economies in the presence of strategic agents; however, their notion of efficiency
was weaker than the PoA.
They termed an outcome $\mu$-efficient if there was no way of improving everyone's outcome
in terms of utility by an additive factor of $\mu$, and showed that with high probability (i.e.\ $1-\mu$)
a $\mu$-efficient outcome would occur when the size of the economy was large enough, so long
as each agent was small, agents were truthful with non-zero probability, and some
additional more technical conditions.
In contrast, the PoA considers the ratio {between the} social welfare {at the competitive equilibrium and} the achieved social welfare, 
namely a ratio of the sum of everyone's outcomes.
\cite{AzevBudish2012} showed that the Strategy Proof in the Large methodology could be applied to
exchange economies for agents that are limited to having a finite type space, 
independent of the size of the economy; in contrast, our results do not require a restriction of this sort.
Finally, we note that our
bounds apply to classes of Fisher markets, 
whereas the earlier literature applies to classes of exchange economies,
which is a significantly more general setting;
nonetheless, there are settings our work handles which are not covered by
these prior works.\footnote{For example, buyers with linear utility functions but an infinite bidding space.}

\cite{BranzeiCDFFZ14} analyzed the PoA of strategizing in Fisher markets.
% In this game bids (valuation descriptions) were treated as the true valuations and used to determine a market equilibrium.
The PoA compared the social welfare of the worst resulting Nash Equilibrium to the optimal,
i.e.\ welfare maximizing assignment, under a suitable normalization of utilities.
Among other results, they showed lower bounds of $\Omega(\sqrt n)$ on the PoA when there are $n$ buyers
with linear utilities.
However, we view the comparison point of an optimal assignment to be too demanding in this
setting, as it may not be an assignment that could arise based on a pricing of the goods.
In our results we will be comparing the strategic outcomes to those that occur under truthful bidding.
Another approach is to bound the gains to individual agents,
called the \emph{incentive ratio}; \cite{ChenDZ2011,ChenDZZ2012} showed these values
were bounded by small constants in Fisher market settings.

There has been much other work on large settings and their behavior.
We mention only a sampling.
\cite{Kalai2004} studied the notion of extensive robustness for large games,
and~\cite{KalaiShmaya2013} investigated large repeated games using the notion of compressed equilibria.
\cite{PaiRU2014} studied repeated games and the use
of differential privacy as a measure of largeness.
In a different direction, \cite{GradwohlReingold2008} investigated fault
tolerance in large games for $\lambda$-continuous and anonymous games.

\section{Preliminaries}
\label{sec:prelim}

\subsection{Definitions for Large Walrasian Auctions}

\begin{definition}
\label{def:wal-equi}
An auction $A$ comprises a set of $N$ bidders $B_1, B_2,\ldots, B_N$,
and a set of $m$ goods $G$, with $n_j$ copies of good $j$, for $1 \le j \le m$.
We write $\bn = (n_1, n_2, ..., n_m)$,
where $n_j$ denotes the number of copies of good $j$,
and we call it the \emph{multiplicity} vector.
We also write $\bn = (n_j, n_{-j})$,
where $n_{-j}$ is the vector denoting the number of copies of goods other than good $j$.
We refer to an instance of a good as an item.
For an allocation $\bxi$ to bidder $i$, which is a subset of the available goods,
we write $\bxi=(x_{i1},x_{i2}, \ldots, x_{im})$ where $x_{ij}$ denotes the number
of copies of good $j$ in allocation $\bxi$.
There is a set of prices $\bp=(p_1, p_2, \ldots, p_m)$, one per good;
we also write $\bp=(p_j, p_{-j})$.
Each bidder $i$ has a valuation function $v_i:X\rightarrow \mathbb{R}_+$,
where $X$ is the set of possible assignments,
and a quasi-linear utility function $u_i(\bxi) =v_i(\bxi) -\bxi \cdot \bp$.

A Walrasian equilibrium is a collection of prices $\bp$ and an allocation $\bxi$
to each bidder $i$ such that (i) the goods are fully allocated but
not over-allocated, i.e.\ for all $j$, $\sum_i x_{ij} \leq n_j$, and $\sum_i x_{ij} = n_j$ if $p_j > 0$,
and (ii) each bidder receives a utility maximizing allocation at prices $\bp$, i.e.\
$u_i(\bxi) = v_i(\bxi) - \bxi \cdot \bp = \max_{\bxi} [v_i(\bxi) - \bxi \cdot \bp]$.

In a Walrasian mechanism for auction $A$ each bidder declares % produces 
a bid function
$b_i: X\rightarrow \mathbb{R}_+$.
We write $\bb = (b_1, b_2, \ldots, b_N)$ and $\bb = (b_i, b_{-i})$.
The mechanism computes prices and allocations
as if the bids were the valuations.
\end{definition}

Given the bidders and their bids,
$\bp(\bn; \bb)$ denotes the prices produced by the Walrasian mechanism at hand
when there are $\bn$ copies of the goods
and $\bb$ is the bidding profile.
Also, $p_j(\bn; \bb)$ denotes the price of good $j$ and
$\bp(\bn; \bb) = (p_j(\bn; \bb), p_{-j}(\bn; \bb))$.
Finally, we let both $\bxi(\bn; \bb)$ and $\bxi(\bn; b_i,b_{-i})$
denote the allocation to bidder $i$ provided by the mechanism.

\begin{definition}
\label{def:gross-sub}
A valuation or bid function satisfies the gross substitutes property if
increasing the price for one good only increases the demand for other goods.
Formally,
for each
utility maximizing allocation $\bx$ at prices $\bp=(p_j, p_{-j})$, at prices
$(q_j, p_{-j})$ such that $q_j > p_j$, there is a utility maximizing allocation $\by$ with $y_{-j} \succeq x_{-j}$
(i.e.\ $ y_k \ge x_k$ for $k \neq j$). This definition applies to the Fisher market setting also.
\end{definition}

In the auctions we consider the number $\bn$ of copies of each good is determined
by a distribution $F(\bn\hide{,N})$. In order for the auction to be large, we need that the probability that there
are exactly $r_j$ copies of the $j$-th item be small, for every $r_j$ and for every $j$.

\hide{
A large Walrasian auction is an auction with $N$ bidders where $N$ is large.
However, in order to state theorems parameterized by $N$,
we define a large auction as being a sequence of auctions with increasing numbers
of bidders, as follows.}

\begin{definition}
\label{large-market}
A large Walrasian auction is characterized by a distribution $F(\bn\hide{,N})$, a demand bound $k$,
and a largeness measure $L$.
% A large Walrasian auction is a sequence of auctions $A_1, A_2, \ldots, A_N, \ldots$,
% where $N$ denotes the number of bidders.
It satisfies the following two properties.
\begin{enumerate}[i.]
\item
The demand of every bidder is for at most $k$ items.
Formally, if allocated a set of more than $k$ items, the bidder will obtain equal utility
with a subset of size $k$.
\item
The probability that there are exactly $c$ copies of good $j$, for any $c$ and any $j$ 
is bounded by $1/L$.
Formally,
Let $F(n_j,j\hide{,N} | n_{-j})$ denote the probability that there are exactly $n_j$ copies of good $j$ when given $n_{-j}$ copies of other goods;
then $\max_j \max_{n_j, n_{-j}} F(n_j,j\hide{,N} | n_{-j})\le 1/L$. % $\lim_{N\ra \infty} F(N) = 0$.
\end{enumerate}
\end{definition} 
\hide{
\rjc{Note that this definition implies that with high probability there are many copies of each good;
e.g., there is at least a 50\% probability of at least $L/(2m)$ copies of every good being present.
This then also implies that to avoid the trivial situation in which all bidders' demands can be optimally
satisfied, there must be at least $L/(2mk)$ bidders at hand.
In other words, if $k$ is small and $L$ is large, then with high probability the market will be large:
it will have
lots of copies of each good, and to avoid the PoA being trivially close to 1, there needs to be
lots of buyers present too.}
}
\hide{It will be convenient to write $m_j= \mu(n_j)$ for the expected number of copies of good $j$.}

Note this definition implies that the expected number of copies of each good is at least $\frac{L}{2}$ and it is in this sense that the market is large.

A Bayes-Nash equilibrium is an outcome with no expected gain from an individual deviation:
\[
\forall b'_{i}: \mathbb{E}_{\bn, v_{-i},b_{-i}}\left[u_i(x_i(\bn; b_i,b_{-i}),p((b_i,b_{-i}))\right] \ge
                      \mathbb{E}_{\bn, v_{-i},b_{-i}}\left[u_i(x_i(\bn; b'_i,b_{-i}),p((b'_i,b_{-i})) \right].
\]

The social welfare $\sw(\bx)$ of an allocation $\bx$ is the sum of the individual valuations:
$\sw(\bx) = \sum_i v_i(\bxi)$.
We also write $\sw(\opt)$ for the (expected) optimal social welfare, the maximum (expected) achievable social welfare,
and $\sw(\ne)$ for the smallest (expected) social welfare achievable at a Bayes-Nash equilibrium.

Finally, the Price of Anarchy is the worst case ratio of $\sw(\opt)$ to $\sw(\ne)$
over all instances in the class of games at hand, namely
% which in this context comprise 
auctions $A_N$ of $N$ buyers:
\[
\text{PoA} = \max_{A_N} \frac{\sw(\opt)} {\sw(\ne)}.
\]

\subsection{Definitions for Large Fisher Markets}

\begin{definition}
\label{def:fisher-market}
A \emph{Fisher market}\footnote{In much of the Computer Science literature the term market has been used to mean what is called an economy in the economics literature.}
has $m$ divisible goods and $N$ agents, called buyers.
There is a fixed endogenous supply of each good (which WLOG is chosen to be 1 unit).
Agent $i$ has a fixed endowment of $e_i$ units of money.
Each agent has a utility function, with the characteristic
that the agent has no desire for its money, i.e.\
each agent seeks to spend all its money on goods.
Suppose we assign a price $p_j$ to each good $j$, then a (possibly non-unique) {\em demand} of agent $i$
is a bundle of goods $(x_{i1}, x_{i2},\ldots, x_{im})$
that maximizes her utility subject to the budget constraint:
\ $  \sum_j p_j  x_{ij}  \leq e_i$.
A \emph{market demand} $\xj$ for a good $j$ is the total (possibly non-unique) demand for that good;
$\xj = \sum_i x_{ij}$.
This is viewed as a function of the price vector $\bp = (p_1,p_2,\ldots,p_m)$.
Prices $\bp$ form a \emph{Walrasian equilibrium}, or equilibrium for short, if the resulting
markets can clear, that is there exists a market demand at these prices such that
for all $j$, $\xj = 1$ if $p_j > 0$.
\hide{For notational convenience, we define an \emph{excess demand} for good $j$ as $\zj = \xj -1$.
The equilibrium condition is that every excess demand be zero.
Prices $\bp$ form an equilibrium if there exists market demands at these prices that clear the market.}
\end{definition}
The Fisher market is actually a special case of an Exchange economy.
(To see this, view the money as another good, and the supply of the goods as being
initially owned by another agent, who desires only money.)

In general computing equilibria is computationally hard even for Fisher markets~\cite{ChenTeng2009,VaziraniYann2011}. One feasible class is the class of Eisenberg-Gale markets, markets for which the equilibrium computation 
becomes the solution to a convex program. This class was named
in~\cite{JainVaz2007};
the program was previously identified in~\cite{EisenbergGale1959}.
\begin{definition}
\label{def-EGmarket}
Eisenberg-Gale markets are those economies for which the equilibria are exactly the solutions to the following convex program, called the Eisenberg-Gale convex program:
\begin{align*}
\max_{\bx}  ~~ & \sum_{i=1}^{n}  e_i \cdot \log(u_i(x_{i1}, x_{i2}, \cdots x_{im})) \numberthis \label{EG-convex}\\
&\text{s.t.} ~~~~ \forall j:~~ \sum_i x_{ij} \leq 1~~
\text{and}~~~~\forall i, j:~~ x_{ij} \geq 0. \nonumber
% &\text{s.t.} ~~~~ \forall j:~~~~ \sum_i x_{ij} \leq 1 \nonumber\\
% & ~~~~~~~~~\forall i, j:~~ x_{ij} \geq 0. \nonumber
\end{align*}

In a Fisher market game, each buyer 
declares % produces 
a bid function $b_i$;
however, her endowment is public knowledge. 
The mechanism computes prices and allocations as if the bids were the valuations. 
The same restrictions will apply to the bid functions and the utilities.
The goal of each buyer is to maximize her utility.

\emph{Notational remark}
The demands are induced by the bids, thus we could write
$u_i(\bxi(b_i,b_{-i}))$, but for brevity we will write this as
$u_i(b_i,b_{-i})$ instead. Also, it will be convenient to write $v_i$ for the truthful bid of $u_i$, yielding the notation $u_i(v_i, v_{-i})$.

\begin{definition}
\label{large-fisher}
The largeness $L$ of a Fisher market is defined to be the ratio $L =\frac{\sum_i e_i} {\max_i e_i}$.
\end{definition}

It is natural to measure the efficiency of outcomes  in the Fisher market game using the objective function \eqref{EG-convex}, or rather its exponentiated form. More specifically, we compare the geometric means of the buyer's utilities weighted by their budgets
at the worst Nash Equilibrium (with bids $b$) and at the market equilibrium (with bids $v$), and namely we call it {\bf Geometric Price of Anarchy(GPoA)}:
\hide{\rjc{Changed min to max in next two formulas.}}
\[
\GPoA(M) = \max_{\text {NE with bids}~ b}
\Bigg(\prod_i \Big( \frac {u_i(v_i, v_{-i})} {u_i(b_i, b_{-i})}        
                      \Big)^{e_i}
\Bigg)^{\frac{1}{\sum_i e_i}}.
\]

Note that in the settings we consider the prices at a market equilibrium are unique.
We will also use this product to upper bound a Price of Anarchy notion for a market $M$, 
which compares the sum of the utilities
at the worst Nash Equilibrium to the sum at the market equilibrium.
\[
\text{PoA}(M) = \max_{\text {NE with bids}~ \bb}
\frac
       {\sum_i u_i(v_i, v_{-i}) } 
       { \sum_i u_i(b_i, b_{-i}) }  . 
\]
For the latter measure to be meaningful, we need to use a common scale for the different buyers' utilities.
To this end, we define \emph{consistent scaling}.
\begin{definition}
\label{def:consistent-util}
The bidders' utilities are consistently scaled if there is a parameter $t > 0$ 
such that for every bidder $i$, $u_i(v_i, v_{-i}) = t e_i$. \footnote{WLOG, we may assume that $t = 1$.}
That is, bidder $i$'s utility function is scaled to give it utility $t e_i$ at the market equilibrium,
where $e_i$ is its budget.
\end{definition}

Finally, we will be considering utility functions that are  monotone, homogeneous of degree one
(defined below),  continuous, concave, and that 
induce demands that satisfy the gross substitutes condition (see Definition~\ref{def:gross-sub}).

\begin{definition}
\label{def:homo-deg-one}
Utility function $u(\bx)$ is \emph{homogeneous of degree 1} if for every $\alpha >0$, 
$u(\alpha \bx) = \alpha \cdot u(\bx)$.
\end{definition}

\begin{fact}
The utility functions in Eisenberg Gale programs are assumed to be homogeneous of degree 1, continuous and concave.
\end{fact}
\end{definition}

\hide{
\begin{definition}
\label{def:wgs-divisible}
The demand of a buyer specified by the relation $\bx(\bp)$ (note that it need not be unique)
satisfies the weak gross substitutes property if for all prices $\bp$ and $\bq$,
with $q_j > p_j$ and $q_k = p_k$ for $k \neq j$,
given any specific demand $\bx(\bp)$,
there is a demand $\bx(\bq)$ such that $x_k(q) \ge x_k(p)$ for all $k\neq j$
and $x_j(q) \le x_j(p)$.
\end{definition}
}

\hide{
As a running example, to illustrate the application of the large auction definition,
we will use a binomial distribution with $2m_j$ potential copies
of good $j$, each with a probability $\tfrac 12$ of being present.
For this binomial distribution, $F(j,N) \le \frac{1}{\sqrt{2m_j}}$.
% Note that there are between $0$ and $2m_j$ copies of good $j$.
}

\subsection{Regret Minimization}
In a regret minimization setting, a single player is playing a repeated game. At each round, the player can choose to play one of $K$ strategies, which are the same from round to round. The outcome of the round is a payoff in the range $[-\chi, \chi]$. 

\begin{definition}
An algorithm that chooses the strategy to play is regret minimizing if the outcome of the algorithm, in expectation, is almost as good as the outcome from always playing a single strategy regardless of  any one else's actions. Formally, there is a function $\Phi(|K|,T) = o(T)$ such that, for any $b_{-i}^t$, for any fixed strategy $b_{i} \in K$,
\[
\sum_{t = 1}^{T} u_i (b_{i}^t, b_{-i}^t) - \sum_{t = 1}^{T} u_i (b_{i}, b_{-i}^t) \geq - \Phi(|K|,T) \cdot \chi, 
\]
where \hide{$\Phi(|K|,T) = o(T)$ and}$b_{i}^t$ is the strategy bidder $i$ uses at time $t$.
\end{definition}

\begin{theorem}
Regret minimizing algorithms exist. If, at the end of each round, the player learns the payoff for all $K$ strategies, $\Phi(|K|,T) = O(\sqrt{T})$ can be achieved, and if she learns just the payoff for her strategy, $\Phi(|K|,T) = O(T^{\frac{2}{3}})$ can be achieved.
\end{theorem}

Note that in large auctions and markets, it is the latter result that seems more applicable.

As shown in \cite{Roughgarden2015}, if all players play regret minimizing strategies, the resulting outcome observes the PoA bound obtained via a smoothness argument up to the regret minimization error.

\section{Our Results}
\label{sec:results}

One issue that deserves some consideration when specifying a large setting, 
and placing some inevitable restrictions on the
possible settings, is to determine which parameters should remain bounded as the setting size grows.
So as to be able to state asymptotic results, we give results in terms of a parameter $L$ which is allowed to grow arbitrarily large. \footnote{The only interesting case in these Walrasian auctions is when $N = \Omega(L)$, for otherwise, due to bounded demands, all buyers achieve their optimal allocation w.h.p.} But in fact all settings are finite, so really when stating that some parameters are bounded, we are making
statements about the relative sizes of different parameters.

One common assumption is that the type space is finite.  
However, it is not clear such an assumption is desirable
in the settings we consider, 
for it would be asserting that the number of possible valuations and bidding strategies is much
smaller than the number of bidders.  Another standard assumption is that the ratio of the
largest to smallest non-zero valuations are bounded. 
This, for example, would preclude valuations being distributed according to a power law distribution
(or any other unbounded distribution), which again seems unduly restrictive if it can be avoided.

\subsection{Result for Walrasian Auctions}

Our analysis makes two assumptions; stronger assumptions were made
for the large auction results in~\cite{Swinkels2001,FeldmanILRS:15}. \cite{Swinkels2001} also ruled out overbidding by arguing it is a dominated strategy. Our analysis can avoid even this assumption of other players' rationality; however, bounded overbidding is needed for the extension to regret minimizing strategies.
\begin{assumption}
\label{ass:bdd-val}
\emph{[Bounded Expected Valuation]~}
There is a constant $\zeta$ such that
for each bidder and each item, her expected value for this single item is at most $\zeta$:
\[
\max_s \mathbb{E}_{v_i} [v_i(s)] \leq \zeta.
\]
\end{assumption}
Note that without this assumption the social welfare would not be bounded, and then it is not clear how to measure the
Price of Anarchy.
Prior work had assumed $v_i(s) \le \zeta$ for all $s$ and $i$ (i.e.\ absolutely rather than in expectation).

\begin{theorem}
\label{wal-e1q-PoA}
In a large Walrasian auction which satisfies Assumption~\ref{ass:bdd-val} and with buyers whose valuation and bid functions are monotone and satisfy the gross substitutes property,
% the Bounded Valuation and Market Size properties,
% If $\text{SW(OPT)} > \rho N$, then
\[
\sw(\ne) \geq \Bigg(1 - \frac{3 \cdot k \cdot (k+1) \cdot \zeta \cdot m}{\rho} \cdot  Y \cdot  \lceil \log_2 \frac{1}{Y}\rceil  \Bigg) \sw(\opt),
\]
where $Y = \frac{m}{L} \left[ 2 m {{k+1+m}\choose{m}} \right]$ and $\rho = \frac{\sw(\opt)}{N}$ .
\hide{
In particular, if the number of copies of each good is independently and identically distributed according to the Binomial distribution $B(N, \frac 12)$, 
and $k,m =O(1)$, then
\[
\sw(\ne) \ge \left( 1 -O\left(\frac {\log N} {\sqrt N} \right) \right) \sw(\opt).
\]}

Also, if there is only one good, i.e., if $m = 1$, then
\[
\sw(\ne) \geq \Bigg(1 - \frac{3 \cdot k \cdot (k+1) \cdot \zeta }{\rho} \cdot  Y \cdot  \lceil \log_2 \frac{1}{Y}\rceil  \Bigg) \sw(\opt),
\]
where $Y =\frac{ 2 (k+2)}{L} $ and $\rho = \frac{\sw(\opt)}{N}$ .
\end{theorem}
{\bf Remark}
The gross substitutes assumption is present so as to ensure the auction outcome is a Walrasian equilibrium w.r.t.\ to the bids,
for if it is not then some bidders will be allocated a non-favorite bundle, which seems unattractive as a solution concept.

To achieve $\sw(\ne) \geq (1 - \epsilon) \sw(\opt)$ where $\epsilon$ is small, we need $\frac{L}{\rho \cdot \log L}$ to be large. 
We can achieve this by considering a sequence of auctions indexed by $N$, the number of bidders,  
and requiring  $\rho$ to be a constant and $L$ to be sufficiently large. 
One way to obtain a constant $\rho$ is to make the following two assumptions.
\hide{
\begin{assumption}
\label{ass:sw-grows}
\emph{[Market Welfare]~}
The optimal social welfare grows linearly with the number of bidders:
$\sw(\opt) \ge \rho N$, for some constant $\rho > 0$.
\end{assumption}
}
\hide{
\cite{FeldmanILRS:15} also makes this assumption. \cite{Swinkels2001} makes assumptions on the value distribution which imply Assumption~\ref{ass:sw-grows} although this consequence is not stated in his work.

We can achieve Assumption~\ref{ass:sw-grows} by making following assumptions.
}
\begin{assumption}
\label{ass:mark-size}
\emph{[Auction Size]~}
Let $\mu(n_j)$ be the expected number of copies of good $j$, for $1\le j \le m$, and let $\Gamma(n_j)$ be its standard deviation.
The assumption is that
for each $j$, $\mu(n_j) = \Theta(N)$ and $\Gamma(n_j) \le (1 - \lambda) \mu(n_j)$ for some constant $\lambda > 0$.
Let $\alpha >0$ be such that $\mu(n_j) \ge \alpha N$ for all $j$ and sufficiently large $N$.
\end{assumption}

\begin{assumption}
\label{ass:min-bound}
\emph{[Value Lower Bound]~}
There is a parameter $\rho'>0$ such that
for any bidder, its largest expected value for one item 
% the minimum valuation of any bidder for any one item set for which it has non-zero value
 is at least $\rho'$:
 \[
\max_s \mathbb{E}_{v_i} [v_i(s)] \geq \rho'.
 \]
\end{assumption}

\begin{lemma}
\label{lem:high-sw}
Let $\rho = \lambda^2 \alpha \frac{2 \lambda + \lambda^2}{(1 + \lambda)^2} \rho'$.
% $\rho = \min\{ \frac 1l, \lambda^2 \alpha\}\frac{2 \lambda + \lambda^2}{(1 + \lambda)^2} \rho'$.
If Assumptions~\ref{ass:mark-size}  and~\ref{ass:min-bound} hold, then $\sw(\opt) \ge \rho N$.
\end{lemma}

In previous work, \cite{FeldmanILRS:15} also made the assumption that $\rho$  is a constant. 
\cite{Swinkels2001} made assumptions on the value distribution which again imply $\rho$ is a constant although this consequence is not stated in his work.

\begin{corollary}
In a large Walrasian auction which satisfies Assumption~\ref{ass:bdd-val} and with buyers whose valuation and bid functions are monotone and satisfy the gross substitutes property, if the number of copies of each good is independently and identically distributed according to the Binomial distribution $B(N, \frac 12)$, 
and $\rho,k,m =O(1)$, then
\[
\sw(\ne) \ge \left( 1 -O\left(\frac {\log N} {\sqrt N} \right) \right) \sw(\opt).
\]
\end{corollary}

In order to obtain good bounds when using regret minimization algorithms, we need to be able to bound
the possible losses a player makes, which we achieve by bounding the possible overbidding.
This is similar to the notion of overbidding previously given in~\cite{BabaioffLNP14}.
\begin{definition}
Let $K$ be the set of strategies a player uses. She is a $(\gamma, \delta)$-player if $v \in K$ and, for any $b \in K$ and for any set $\bx$, 
\[
b(\bx) \leq v(\bx) \cdot  \gamma + \delta.
\]
\end{definition}
\begin{theorem}
\label{thm:regret-wal}
Suppose all players use regret minimization algorithms, they are all $(\gamma, \delta)$-players and their valuation and bid functions are monotone and satisfy the gross substitutes property. Then, in a large Walrasian auction which satisfies Assumption~\ref{ass:bdd-val},
\begin{align*}
\frac{1}{T}\mathbb{E}_{\bn, \bv, \bb}\Bigg[\sum_{t = 1}^{T} v_i(\bxi(b_i^t, b_{-i}^t))\Bigg] &\geq \Bigg(1 - \frac{3 \cdot k \cdot (k+1) \cdot \zeta \cdot m}{\rho} \cdot  Y \cdot  \lceil \log_2 \frac{1}{Y}\rceil  \\
&~~~~~~~~~~~~~~~~~~~~~~- \frac{\max_{i} \Phi(|K_i|,T) \cdot (k m \zeta \gamma + \delta)}{\rho \cdot T} \Bigg) \sw(\opt).
\end{align*}
where $Y = \frac{m}{L} \left[ 2 m {{k+1+m}\choose{m}} \right]$,  $\rho = \frac{\sw(\opt)}{N}$, $K_i$ is the set of strategies used by $i$, and $v_i \in K_i$.
\end{theorem}

\subsection{Fisher Market Results} 
\begin{theorem}
\label{thm:FisherPoA}
Let $M$ be a large Fisher market with largeness $L$
in which the utility and bid functions are homogeneous of degree 1, concave, continuous, monotone and satisfy the gross
substitutes property.
If its demands as a function of the prices are unique at any $\bp > 0$, or if its utility functions are linear,
then its Price of Anarchy and its Geometric Price of Anarchy is bounded by
\begin{align*}
\PoA(M) \le e^{m/L},
\GPoA(M) \le e^{m/L},
\end{align*}
where $m$ is the number of distinct goods in the market.
\end{theorem} 

Perhaps surprisingly, there is no need for uncertainty in this setting.
Note that these assumptions on the utility functions are satisfied by Cobb-Douglas utilities,
and by those CES and Nested CES utilities that meet the weak gross substitutes condition.
We note that our results do not extend to Fisher markets with Leontief utilities.
For Theorem 4 in~\cite{BranzeiCDFFZ14}
can be readily adapted to show that for some Fisher markets with Leontief utility functions, 
when $L$ is large, the PoA is at least $m$, the number of goods.

\begin{theorem}
\label{thm:regret-fisher}
Suppose all players use regret minimization algorithms, their utility functions and bid functions are homogeneous of degree 1, concave, continuous, monotone, and satisfy the gross
substitutes property.
If its demands as a function of the prices are unique at any $\bp > 0$, or if its utility functions are linear, then in a  large Fisher Market with largeness $L$ and with reserve prices $\mathbf{r}$ such that for any $j$, $\frac{1}{\lambda} \leq \frac{r_j}{p_j(\bv)} \leq \frac{1}{4}$, 
\[
\frac{1}{T} \sum_{t = 1}^{T} \sum_{i} u_i (b_{i}^t, b_{-i}^t) \geq  (e^{-\frac{2 m}{L}} - \frac{max_i \Phi(|K_i|, T)}{T} \lambda )\sum_{i} u_i (v_{i}, v_{-i}),
\]
where $K_i$ is the set of strategies used by player $i$ and $v_i \in K_i$.
\end{theorem}
\section{Large Walrasian Auctions}
\label{sec:wal-equil}

Here we prove a slightly weaker version of Theorem~\ref{wal-e1q-PoA} which demonstrates the main ideas
(Theorem~\ref{SPL-main-idea} below). 
Our goal is to show that in expectation
\[
\sum_i \mathbb{E}\Big[u_i(x_i(v_i, b_{-i}))\Big]
\geq  \text{SW(OPT)} - \text{R} - O(N\eps),
\]
where $R$ denotes the expected auction revenue.
For then we can apply the smooth technique for Bayesian
settings~\cite{Syrgkanis:2013:CEM:2488608.2488635} to obtain our result\hide{,
assuming $\sw(\opt) = \Omega(N)$ (Assumption~\ref{ass:sw-grows})}.

This follows from two observations.
First, with high probability, a buyer has at most a small influence on prices (Lemma~\ref{k-bad probability}), and hence can improve her own utility by at most a small amount via a non-truthful bid (Lemma~\ref{smooth}).
Otherwise, by Assumption~\ref{ass:bdd-val} and the Gross Substitutes property,
her expected utility is bounded by $k m \zeta$.
\hide{ (Lemma~\ref{1-to-l valuation}).
Perhaps we should omit the last lemma as it is rather obvious.
}
The probability bound stems from the distribution $F$ over the number of goods.
To obtain the bound, 
we define $(k,\epsilon)$-good and bad multiplicity vectors $\bn$, wr.t.\ bids $\bb$.
By counting their number, we will show that the fraction of $(k+1,\epsilon)$-bad vectors is
$O(\frac{1}{L\epsilon})$.
Also, if the vector is $(k+1,\epsilon)$-good, we will show that a bidder
can cause the prices, when they are all bounded by 1, to vary by at most $(k+1) \epsilon$.
Essentially,
a vector $\bn$ is $(k,\epsilon)$-good if changing the supplies by
at most $k$ items causes prices $p_j\le 1$ to change in total by at most $k\epsilon$. 
Then, using the fact that the equilibrium is Walrasian, we can show that
for $(k+1,\epsilon)$-good vectors $\bn$,
\[u_i(x_i(v_i, b_{-i})) \geq v_i(x_i(v_i, v_{-i})) - \sum_{s \in x_i(v_i, v_{-i})} p_s(\bn; (b_i, b_{-i})) - k(k+1) \epsilon. \]
On summing over $i$ and taking expectations, we can then deduce that
\[
\sum_i \mathbb{E}\Big[u_i(x_i(v_i, b_{-i}))\Big]
\geq  \text{SW(OPT)} - \text{R} - N \cdot k \cdot (k + 1) \cdot \epsilon - O(\frac{N \cdot k}{L\epsilon}).
\]

\hide{By Assumption~\ref{ass:sw-grows}, $\sw(\opt) = \Theta(N)$.
We can now apply the smooth technique for Bayesian
settings~\cite{Syrgkanis:2013:CEM:2488608.2488635} to obtain our result.
}
%\newline

\smallskip

Recall that the English Walrasian mechanism can be implemented as an ascending auction.
The prices it yields can be computed as follows:
$p_j$ is the maximum possible increase in the social welfare when
the supply of good $j$ is increased by one unit.
Similarly, the Dutch Walrasian mechanism can be implemented as a descending auction,
and the resulting price $p_j$ is the loss in social welfare when
the supply of good $j$ is decreased by one unit.

We will be considering an arbitrary Walrasian mechanism.
Necessarily, its prices must lie between those of the Dutch Walrasian and English Walrasian mechanisms.
We let $\bp^{Eng}(\bn; (b_i, b_{-i}))$ denote the price output by the English Walrasian mechanism and $\bp^{Dut}(\bn; (b_i, b_{-i}))$ be the price output by the Dutch Walrasian mechanism.

We define the distance between two price vectors $\bp$ and $\bp'$ with respect to $U$ as follows:
\[
dist^U(\bp,\bp' ) = \sum_{j=1}^m \left | \min\{p_j,U\} - \min\{p'_j, U\} \right | .
\]

\begin{definition}
Given bidding profile $(b_i, b_{-i})$,
$\bn = (n_j, n_{-j})$ is $(\epsilon, U)$-\emph{bad} for good $j$
if,  in the English Walrasian mechanism,
the distance between the prices is more than $\epsilon$ when an additional copy of good $j$ is added to the market:
\[
dist^U(\bp^{Eng}((n_j, n_{-j}); (b_i, b_{-i})), \bp^{Eng}((n_j + 1, n_{-j}); (b_i, b_{-i}))) > \epsilon.
\]
\end{definition}
Let $\bk = (k,k,\ldots, k)$ and $\mathbf{0} = (0,0, \ldots, 0)$ be $m$-vectors.
\begin{definition}
Given bidding profile $\bb$,
$\bn$ is $(k, \epsilon, U)$-\emph{bad} for good $j$
if there is a vector $\bn'$ which is $(\epsilon, U)$-bad for good $j$,
% new definition
such that $n'_h \le n_h$ for all $h$, and
$\sum_h n_h \le k+ \sum_h n'_h$.
% and which is less than $\vec n$ by at most $k$ units in each dimension: $\{\vec 0, \vec n' - \vec k\} \preceq \vec n' \preceq \vec n$.
$\bn$ is $(k, \epsilon, U)$-\emph{good} if it is not $(k, \epsilon, U)$-bad.
\end{definition}
\hide{
In Lemmas~\ref{epsilon-bad} and~\ref{k-epsilon-bad}, we bound the number of $(\epsilon,U)$-bad multiplicity
vectors, and then in Lemma~\ref{k-bad probability} we bound the probability of a $(k, \epsilon, U)$-bad vector. Following this, in Lemma~\ref{lowboundprice} and \ref{uppboundprice},
assuming the multiplicity vector is $(k+1, \epsilon, U)$-good,
we bound the difference between the English Walrasian mechanism prices and those of the Walrasian mechanism at hand. Next, in Lemma~\ref{smooth}, again for $(k+1,\epsilon, U)$-good multiplicity vectors, we relate $u_i(\bxi(v_i, b_{-i}))$ to $v_i(\bxi(v_i, v_{-i}))$ and the prices paid; we then use this to carry out a PoA analysis.
}
\hide{
In Lemma~\ref{k-bad probability} we bound the probability of a $(k, \epsilon, U)$-bad vector. Following this, in Lemma~\ref{smooth}, for $(k+1,\epsilon, U)$-good multiplicity vectors, we relate $u_i(\bxi(v_i, b_{-i}))$ to $v_i(\bxi(v_i, v_{-i}))$ and the prices paid; we then use this to carry out a PoA analysis. }
For brevity, we sometimes write $u_i(v_i, b_{-i})$ instead of $u_i(\bxi(v_i, b_{-i}))$.
For simplicity, let $\Lambda(m, k)$ denote $m\cdot{{ k+m} \choose m}$.
\hide{
\begin{lemma}\label{epsilon-bad}
In the English Walrasian mechanism, given $n_{-j}$ and bidding profile $\bb$, the number of values $n_j$ for which $(n_j, n_{-j})$ is \textit{$(\epsilon, U)$-bad} for good $j$ is at most $\frac{m}{\epsilon}U$.
\hide{
\[
\left | \left \{n_j | dist(p^{Eng}((n_j, n_{-j}); \vec b), p^{Eng}((n_j + 1, n_{-j}); \vec b)) > \epsilon \right \} \right | \leq \frac{l}{\epsilon}.
\]
}
\end{lemma}
}

\hide{
\begin{proof}
We prove the result by contradiction. Accordingly,
let\\
$S = \left \{n_j | dist(p^{Eng}((n_j, n_{-j}); \vec b), p^{Eng}((n_j + 1, n_{-j}); \vec b)) > \epsilon \right \}$ and suppose that $|S| > \frac{l}{\epsilon}$.
\[
\text{The proof uses a new function}~\mathit{pf}(\cdot):~~~~
\mathit{pf}(n_j) = \sum_{q = 1}^l \min\{p^{Eng}_l((n_j, n_{-j}); \vec b),1\}.
\]
\begin{eqnarray} \label{label1}
\mbox{Then,} &&\liminf_{n \rightarrow \infty} (\mathit{pf}(0) - \mathit{pf}(n)) = \liminf_{n \rightarrow \infty} \sum_{h = 0}^{n-1} (\mathit{pf}(h) - \mathit{pf}(h+1)) \nonumber\\
&&~~~~~~\geq \sum_{n_j \in S}(\mathit{pf}(n_j) - \mathit{pf}(n_j + 1)) > \frac{ l}{\epsilon} \epsilon = l.
\end{eqnarray}
The first inequality follows for by Observation \ref{non-increasing}, $\mathit{pf}(\cdot)$ is a non-increasing function.
Further, by construction, $0 \leq \mathit{pf}(h) \leq l$ for all $h$, thus $\liminf_{n \rightarrow \infty} (\mathit{pf}(0) - \mathit{pf}(n)) \leq l$, contradicting (\ref{label1}).
\end{proof}
}

\hide{
We say that $\vec n$ is \textit{$(k, \epsilon)$-bad} for good $j$ with bidding profile $(b_i, b_{-i})$, if and only if for some $\vec n'$, $\vec 0, \vec k \prreceq \cecn' \preceq \vec n$,
$\vec n'$ is $\epsilon$-bad for good $j$ with bidding profile $(b_i, b_{-i})$.
We say that $\vec n$ is $(k, \epsilon)$-good if it is not $(k, \epsilon)$-bad.
}
\hide{
\begin{lemma}\label{k-epsilon-bad}
In the English Walrasian mechanism with bidding profile $\bb$, for a fixed $n_{-j}$, the number of values $n_j$ for which $(n_j, n_{-j})$ is $(k, \epsilon, U)$-bad for good $j$ is at most $\frac{m}{\epsilon}U \cdot {{k+m} \choose m}$.
\end{lemma}
}

\hide{
\begin{proof}
If  $(n_j, n_{-j})$ is $(k, \epsilon)$-bad for good $j$, then there must be a point $(n'_j , n'_{-j})$ which is $\epsilon$-bad for good $j$,  with $n_q - k \leq n'_q \leq n_q$, for all $1 \leq q \leq l$. Fixing $n_{-j}$, there are $(k+1)^{l-1}$ possible values of $n'_{- j}$. For each $n'_{-j}$, by Lemma \ref{epsilon-bad}, there are at most $\frac{l}{\epsilon}$ points that are $\epsilon$-bad for good $j$.
Finally, for a fixed $n'_j$, there are $k+1$ possible choices of $n_j$. Thus, given $n_{-j}$, there are at most $(k+1)^{l-1} \frac{l}{\epsilon} (k+1) =  \frac{l}{\epsilon}(k+1)^l$ $(k, \epsilon)$-bad points for  good $j$.
\end{proof}
}

\begin{lemma}\label{k-bad probability}
In the English Walrasian mechanism with bidding profile $\bb$, 
the probability that $\bn$ is  $(k, \epsilon, U)$-bad for some good, or $\min_{j} n_j \le k$ is at most
\[\frac mL \left[ \frac{U}{\epsilon} \Lambda(m, k) + k+1\right]. \]
% \[\frac{m}{L} \left[ \frac{U}{\epsilon} \Lambda(m, k) + k+1\right]. \]
% which converges to $1$ as $N \rightarrow \infty$.
\end{lemma}
\hide{
\begin{proof}
Conditioned on the bidding profile being $\vec b$,
\begin{align*}
% &\Pr[(\textit{ $\vec n$ is  $(k, \epsilon)$-good for all goods}) \cap (min_{j} \{n_j\}  > k)] \\
&\sum_{1 \leq j \leq l} \Pr[(\textit{ $\vec n$ is  $(k, \epsilon)$-bad for good $j$}) \cup (n_j \leq k)]) \\
&~~~~~~\leq \sum_{1 \leq j \leq l} \Pr[(\textit{$\vec n$ is $(k, \epsilon)$-bad for good $j$})] + \Pr[(n_j \leq k)] \\
&~~~~~~\leq  \sum_{1 \leq j \leq l} \sum_{\overline{n}_{-j}}(\Pr[(\textit{ $\vec n$ is  $(k, \epsilon)$-bad for good $j$}) | n_{-j} = \overline{n}_{-j}] \\
&~~~~~~~~~~~~~~~~~~~~~~~~+ \Pr[(n_j \leq k) | n_{-j} = \overline{n}_{-j}]) \cdot \Pr[n_{-j} = \overline{n}_{-j}] \\
&~~~~~~\leq  \sum_{1 \leq j \leq l} F(j,N) \left[\frac{ l}{\epsilon} (k+1)^l  + k+1\right] ~~~~~~ \mbox{(by Lemma \ref{k-epsilon-bad}).}
\end{align*}
%
% Since $\lim_{N \rightarrow \infty} f_j(N) = 0$ for all $i$, the probability converges to $1$ as $N \rightarrow \infty$.
\end{proof}
}

\hide{
Let $ n^i_j(b_i, b_{-i}) $ denote the number of copies of good $j$ that bidder $i$ receives with bidding profile $(b_i, b_{-i})$ and $n^i(b_i, b_{-i})$ denote the corresponding vector.
% Let $\vec k = (k,k,\cdots,k)$ be an $l$-vector.
Also, let $p^{Eng}(\bn; b_{-i})$ denote the market equilibrium prices when bidder $i$ is not
present.
}
\hide{
\begin{lemma}\label{lowboundprice}
$p^{Eng}_j(\bn; b_{-i}) \leq  p_j(\bn; (b_i, b_{-i})).$
\end{lemma}
}
\hide{
\begin{proof}
Consider the situation with $\vec n' = \vec n -  n^i (b_i, b_{-i})$ and suppose that agent $i$ is not present. Then $p_j(\vec n; (b_i, b_{-i}))$ is a market equilibrium.
\begin{align*}
\text{So} \hspace*{1in}& \forall j ~~~ p^{Eng}_j(\vec n'; b_{-i}) \leq  p_j(\vec n; (b_i, b_{-i})). \\
\text{Since}~\vec n \geq \vec n',\hspace*{0.3in} &\forall j ~~~ p^{Eng}_j(\vec n; b_{-i}) \leq p^{Eng}_j(\vec n'; b_{-i}).\hspace*{2in}
\end{align*}

The lemma follows on combining these two inequalities.
\end{proof}
}

\hide{
\begin{lemma}\label{uppboundprice}
If $\bn$ is $(k+1, \epsilon, U)$-good for all goods, and $n_j > k + 1$ for all $j$,  then
\[ \forall j~~~\min\{p_j(\bn; (v_i, b_{-i})), U\}  \leq \min\{p_j(\bn; (b_{i}, b_{-i})), U\} + (k + 1) \epsilon.\]
\end{lemma}
}
\hide{
\begin{proof}
First, if $n^i_j(v_i, b_{-i}) > d^i_j$ then $p_j(\vec n; (v_i, b_{-i})) = 0$, as the pricing is given by a Walrasian Mechanism.

Consider the scenario with $\vec n'$ copies of goods on offer, where for all $j$, $n'_j =  n_j -  d^i_j $ and suppose that bidder $i$ is not present;  then $p(\vec n; (v_i, b_{-i}))$ is a market equilibrium.
\begin{align*}
&\text{So,} &p_j(\vec n; (v_i, b_{-i})) \leq p^{Dut}_j(\vec n'; b_{-i}).\\
&\mbox{For all $j' \neq j$, let $n''_{j' } = n'_{j' }$ and let $n''_j = n'_j - 1$; then}~~~~~~ & p^{Dut}_j(\vec n'; b_{-i}) \leq p^{Eng}_j(\vec n''; b_{-i}),\\
&\text{and by Lemma~\ref{lowboundprice},}\hspace*{1.05in}&p^{Eng}_j(\vec n''; b_{-i})\leq p^{Eng}_j(\vec n''; {b_i, b_{-i}}).
\end{align*}
As $\vec n$ is $(k+1, \epsilon)$-good for all goods, and as $\sum_h \vec n_h - \vec n''_h \leq k + 1$,  we conclude that
\begin{align}
&~~~\min\{p_j(\vec n; (v_i, b_{-i})),1\} \leq \min\{p^{Eng}_j(\vec n''; (b_i, b_{-i})),1\} \nonumber\\
 &~~~~~~\leq \min\{p^{Eng}_j(\vec n; (b_i, b_{-i})),1\} + (k + 1) \epsilon \leq \min\{p_j(\vec n; (b_i, b_{-i})),1\} + (k + 1) \epsilon.
\end{align}
\end{proof}
}
\hide{
We let $x_i(b_i, b_{-i})$ denote the set of items that bidder $i$ receives with bidding profile $(b_i, b_{-i})$ and  $v_i(x_i(b_i, b_{-i}))$ denote bidder $i$'s value when she gets allocation $x_i(b_i, b_{-i})$.
}
 Let $\left | \bxi(\cdot) \right |$ denotes the total number of items in allocation $\bxi$.
Let $\bdi  \preceq \bxi(v_i, v_{-i})$ be a minimal set with $v_i(\bdi) = v_i(\bxi(v_i,  v_{-i}))$. By Definition~\ref{large-market}(i), $ |\bdi| \leq k$.

\begin{lemma}\label{smooth}
If $\bn$ is $(k+1, \epsilon, U)$-good, where $U \geq v_i(s)$ for every single item $s$, $n_j > k + 1$ for all $j$,
and $v_i$ and $b_i$ satisfy the gross substitutes property for all $i$, then
\[u_i(v_i, b_{-i}) \geq v_i(\bxi(v_i, v_{-i})) - \sum_{s \in x_i(v_i, v_{-i})} p_s(\bn; (b_i, b_{-i})) - \left | \bxi(v_i, v_{-i}) \cap \bdi \right | \cdot (k+1) \epsilon, \]
where the sum is over all the items in allocation $\bxi$.
\end{lemma}

\hide{
\begin{lemma}\label{1-to-l valuation}
If Assumption~\ref{ass:bdd-val} holds, then $\mathbb{E}_{v_i}[\max_s\{v_i(s)\}] < m \cdot \zeta$.
\end{lemma}
\begin{proof}
$\mathbb{E}_{v_i}[\max_s\{v_i(s)\}] \leq \mathbb{E}_{v_i}[\sum_s v_i(s)] \leq \sum_s \mathbb{E}_{v_i}[v_i(s)] \leq m \cdot \zeta$.
\end{proof}
}

%\vspace{4mm}
\hide{
\begin{proof}
As we are using a Walrasian mechanism, for any allocation $x'_i$,
\begin{eqnarray}
\label{eqn:wal-value}
v_i(x_i(v_i, b_{-i})) - \sum_{s \in x_i(v_i, b_{-i})} p_s(\vec n; (v_i, b_{-i}))
\geq v_i(x'_i) - \sum_{s \in x'_i } p_s(\vec n; (v_i, b_{-i})).
\end{eqnarray}

We let $S$ denote the set of goods whose prices $p_s(\vec n; (v_i, b_{-i}))$ are larger than 1.  Then,
\begin{eqnarray} \label{smooth1}
u_i(v_i, b_{-i}) &=& v_i(x_i(v_i, b_{-i})) - \sum_{s \in x_i(v_i, b_{-i})} p_s(\vec n; (v_i, b_{-i})) \nonumber\\
&\geq& v_i((x_i(v_i, v_{-i}) \cap d_i) \setminus S) - \sum_{s \in (x_i(v_i, v_{-i}) \cap d_i) \setminus S} p_s(\vec n; (v_i, b_{-i}))~~~~\text{(by~\eqref{eqn:wal-value})}
\end{eqnarray}

Since $\vec n$ is $(k+1, \epsilon)$-good, by Lemma~\ref{uppboundprice},
\[\min\{p_s(\vec n; (v_i, b_{-i})),1\} \leq \min\{p_s(\vec n; (b_i, b_{-i})),1\} + (k + 1) \epsilon.\]
\begin{align*}
\text{Therefore, for any}~ s \notin S,~~~~~
p_s(\vec n; (v_i, b_{-i})) &\leq \min\{p_s(\vec n; (b_i, b_{-i})),1\} + (k + 1) \epsilon \\
&\leq p_s(\vec n; (b_i, b_{-i})) + (k + 1)\epsilon.
\end{align*}
\begin{eqnarray}\label{smooth2}
\text{So,}~~~~&&v_i((x_i(v_i, v_{-i}) \cap d_i) \setminus S) - \sum_{s \in (x_i(v_i, v_{-i})  \cap d_i) \setminus S} p_s(\vec n; (v_i, b_{-i})) \nonumber\\
&&~~~~~~\geq v_i((x_i(v_i, v_{-i}) \cap d_i) \setminus S) - \sum_{s \in (x_i(v_i, v_{-i}) \cap d_i) \setminus S} p_s(\vec n; (b_i, b_{-i})) \nonumber\\
&&~~~~~~~~~~~~ - \left | (x_i(v_i, v_{-i}) \cap d_i) \setminus S \right | \cdot (k+1) \epsilon.
\end{eqnarray}

For any $s \in S$, on applying Lemma~\ref{uppboundprice},
we obtain $1 = \min\{p_s(\vec n; (v_i, b_{-i})),1\} \leq \min\{p_s(\vec n; (b_i, b_{-i})),1\} + (k + 1) \epsilon$, which implies $p_s(\vec n; (b_i, b_{-i})) + (k + 1) \epsilon \geq 1$. Also,
\[
v_i(x_i(v_i, v_{-i}) \cap d_i) -  v_i((x_i(v_i, v_{-i}) \cap d_i) \setminus S) \leq
v_i((x_i(v_i, v_{-i}) \cap d_i) \cap S) \leq |(x_i(v_i, v_{-i}) \cap d_i) \cap S|,
\]
where the first inequality follows by the Gross Substitutes assumption,
and the second by Assumption~\ref{ass:bdd-val}.
Thus,
\begin{align}\label{smooth3}
&v_i((x_i(v_i, v_{-i}) \cap d_i) \setminus S) - \sum_{s \in (x_i(v_i, v_{-i}) \cap d_i) \setminus S} p_s(\vec n; (b_i, b_{-i})) - \left | (x_i(v_i, v_{-i}) \cap d_i) \setminus S \right | \cdot (k+1) \epsilon \nonumber\\
&~~~\geq v_i(x_i(v_i, v_{-i}) \cap d_i) - |(x_i(v_i, v_{-i}) \cap d_i) \cap S| \nonumber \\
& \hspace*{1.5in}- \sum_{s \in (x_i(v_i, v_{-i}) \cap d_i)\setminus S} p_s(\vec n; (b_i, b_{-i})) - \left | (x_i(v_i, v_{-i}) \cap d_i) \setminus S  \right | \cdot (k+1) \epsilon  \nonumber\\
&~~~\geq v_i(x_i(v_i, v_{-i}) \cap d_i) - \sum_{s \in x_i(v_i, v_{-i}) \cap d_i} p_s(\vec n; (b_i, b_{-i})) - \left | x_i(v_i, v_{-i}) \cap d_i  \right | \cdot (k+1) \epsilon  \nonumber\\
&~~~\geq v_i(x_i(v_i, v_{-i})) - \sum_{s \in x_i(v_i, v_{-i})} p_s(\vec n; (b_i, b_{-i})) - \left | x_i(v_i, v_{-i}) \cap d_i  \right | \cdot (k+1) \epsilon.  \nonumber\\
\end{align}

By (\ref{smooth1}), (\ref{smooth2}) and (\ref{smooth3}),
\begin{eqnarray}
u_i(v_i, b_{-i}) &\geq& v_i(x_i(v_i, v_{-i})) - \sum_{s \in x_i(v_i, v_{-i})} p_s(\vec n; (b_i, b_{-i})) - \left | x_i(v_i, v_{-i}) \cap d_i  \right | \cdot (k+1) \epsilon. \nonumber
\end{eqnarray}
\end{proof}
}

\hide{
\begin{theorem}
If $\text{SW(OPT)} > \rho N$, then
\[
\text{SW(NE)} \geq  (1 - \frac{k(k+1)\epsilon}{\rho} - \frac{k  \cdot  \sum_{1 \leq j \leq l} f_j(N) ( \frac{l}{\epsilon} (k+2)^l+ k+2)}{\rho})  \text{SW(OPT)} \nonumber
\]
\end{theorem}
}

% Here we prove a slightly weaker version of Theorem~\ref{wal-e1q-PoA} which demonstrates the main ideas. 
% The proof of Theorem~\ref{wal-e1q-PoA} can be found in appendix.
\begin{theorem}
\label{SPL-main-idea}
In a large Walrasian auction which satisfies Assumption~\ref{ass:bdd-val} and with buyers whose valuation and bid functions are monotone and satisfy the gross substitutes property,
% the Bounded Valuation and Market Size properties,
% If $\text{SW(OPT)} > \rho N$, then
\[
\sw(\ne) \geq \left(1 - \frac{3k \cdot \zeta \cdot m}{\rho}\sqrt{ (k+2) \frac{m}{L}  \Lambda(m, k+1) }\right) \sw(\opt), \nonumber
\]
where $\rho = \frac{\sw(\opt)}{N}$.
\end{theorem}
\begin{proof}
By Lemma \ref{smooth},  if $\bn$ is $(k+1, \epsilon \cdot \max_s\{v_i(s)\}, \max_s\{v_i(s)\})$-good and $n_j > k + 1$ for all $j$, then
\begin{eqnarray}
u_i(v_i, b_{-i}) &\geq& v_i(\bxi(v_i, v_{-i})) - \sum_{s \in \bxi(v_i, v_{-i})} p_s(\bn; (b_i, b_{-i})) \nonumber\\
&&~~~~~~~~~~~~~~~~~ - \left | \bxi(v_i, v_{-i}) \cap \bdi  \right | \cdot (k+1) \epsilon \cdot \max_s\{v_i(s)\}. \nonumber
\end{eqnarray}

By Lemma~\ref{k-bad probability}, the probability that $\bn$ is $(k+1, \epsilon \cdot \max_s\{v_i(s)\}, \max_s\{v_i(s)\})$-bad or $n_j \leq k+1$ for some $j$ is less than
\[
\frac{m}{L} \left[ \frac{1}{\epsilon}\Lambda(m, k+1) + k+2\right],
\]
\begin{align*}
\text{and}~~~~\mathbb{E}_{\bn}[u_i(v_i, b_{-i}) ] 
\geq ~& \mathbb{E}_{\bn}\Bigg[ v_i(\bxi(v_i, v_{-i})) - \sum_{s \in \bxi(v_i, v_{-i})} p_s(\bn; (b_i, b_{-i})) \\
&~~~~~~~~~~~~~~~- \left | \bxi(v_i, v_{-i})  \cap \bdi \right | \cdot (k+1) \epsilon \cdot \max_s\{v_i(s)\} \Bigg] \\
&~~~~~ - \frac{m}{L} \left[ \frac{1}{\epsilon}\Lambda(m, k + 1) + k+2\right] \cdot  k \cdot \max_s\{v_i(s)\}. 
\end{align*}
Here, the expectation is taken over the randomness on the multiplicities of the goods; the inequality holds since $u_i(v_i, b_{-i}) \geq 0$ and $v_i(\bxi(v_i, v_{-i})) \leq k \cdot \max_s\{v_i(s)\} $.

Taking the expectation over the valuation of agent $i$ yields
\begin{eqnarray}
\mathbb{E}_{v_i}[\mathbb{E}_{\bn}[u_i(v_i, b_{-i})]] &\geq& \mathbb{E}_{v_i}\Bigg[\mathbb{E}_{\bn}\Big [v_i(\bxi(v_i, v_{-i})) - \sum_{s \in \bxi(v_i, v_{-i})} p_s(\bn; (b_i, b_{-i})) \nonumber\\
~~~~~~~~~~~~~~~~~~~~~~~~~~~~~~~~~~~~~&&~~~~~~~~~~~~~~~~~~~- \left | \bxi(v_i, v_{-i})  \cap \bdi \right | \cdot (k+1) \epsilon \cdot \max_s\{v_i(s)\} \Big ] \nonumber\\
&&~~~ - \frac{m}{L} \left[ \frac{1}{\epsilon}\Lambda(m, k + 1) + k+2\right] \cdot  k \cdot \max_s\{v_i(s)\} \Bigg] \nonumber
\end{eqnarray}

\begin{eqnarray}
&\geq& \mathbb{E}_{v_i}\Bigg[\mathbb{E}_{\bn}\Big[v_i(\bxi(v_i, v_{-i})) - \sum_{s \in \bxi(v_i, v_{-i})} p_s(\bn; (b_i, b_{-i})) \Big]\Bigg]\nonumber\\
&&~~~ - \mathbb{E}_{v_i}[\max_s\{v_i(s)\}]\cdot k \cdot (k+1) \epsilon \nonumber\\
&&~~~ - \mathbb{E}_{v_i}[\max_s\{v_i(s)\}] \frac{m}{L} \left[ \Lambda(m, k + 1) \frac{k}{\epsilon} + (k + 2) k\right]. \nonumber
\end{eqnarray}

Since $\mathbb{E}_{v_i}[\max_s\{v_i(s)\}] \leq \mathbb{E}_{v_i}[\sum_s v_i(s)] \leq \sum_s \mathbb{E}_{v_i}[v_i(s)] \leq m \cdot \zeta$,
\begin{eqnarray}
\mathbb{E}_{v_i}[\mathbb{E}_{\bn}[u_i(v_i, b_{-i})]] &\geq& \mathbb{E}_{v_i}\Bigg[\mathbb{E}_{\bn}\Big[v_i(\bxi(v_i, v_{-i})) - \sum_{s \in \bxi(v_i, v_{-i})} p_s(\bn; (b_i, b_{-i})) \Big] \Bigg]\nonumber\\
&&~~~ - \zeta \cdot m \cdot k(k+1) \epsilon \nonumber\\
&&~~~ - \zeta \cdot m \cdot m \cdot \frac{1}{L} \left[ \Lambda(m, k + 1)\frac{k}{\epsilon} + (k + 2) k \right]. \nonumber
\end{eqnarray}
Let $\text{R}(\bb)$ denote the expected revenue when the bidding profile is $\bb$.
Also, recall that $\sw(\opt) = \rho N$. Now, summing over all the bidders yields

\begin{eqnarray}
\sum_i \mathbb{E}_{\bv, \bb, \bn}[u_i(v_i, b_{-i})] &\geq& \sum_i \mathbb{E}_{\bv, \bb}\Bigg[\mathbb{E}_{\bn}\Big[v_i(\bxi(v_i, v_{-i})) - \sum_{s \in \bxi(v_i, v_{-i})} p_s(\bn; (b_i, b_{-i})) \Big] \Bigg]\nonumber\\
&&~~~ - \zeta \cdot m \cdot m \cdot \frac{1}{L} \left[ \Lambda(m, k + 1) \frac{k}{\epsilon} + (k + 2) k \right]\cdot N \nonumber\\
&&~~~ - \zeta \cdot m \cdot k(k+1) \epsilon \cdot N \nonumber\\
&\geq& \Bigg(1  -\frac{\zeta \cdot m \cdot m \cdot \frac{1}{L} \left[ \Lambda(m, k + 1) \frac{k}{\epsilon} + (k + 2) k\right]}{\rho} \nonumber\\
&&~~~~~~~~~ - \frac{\zeta \cdot m \cdot k(k+1)\epsilon}{\rho}\Bigg)  \text{SW(OPT)} - \mathbb{E}_{\bb}[\text{R}(b_i, b_{-i})]. \nonumber
\end{eqnarray}

Using the smooth technique for Bayesian settings~\cite{Syrgkanis:2013:CEM:2488608.2488635} yields
\begin{eqnarray}
\text{SW(NE)}  &\geq& \Bigg(1  -\frac{\zeta \cdot m \cdot m \cdot \frac{1}{L} \left[ \Lambda(m, k + 1) \frac{k}{\epsilon} + (k + 2) k\right]}{\rho} \nonumber\\
&&~~~~~~~~~ - \frac{\zeta \cdot m \cdot k(k+1)\epsilon}{\rho}\Bigg) \sw(\opt). \nonumber
\end{eqnarray}

Now  set $\epsilon = \sqrt{ \frac{\frac{m}{L} \Lambda(m, k+1) }{k + 1}}$. The claimed bound follows.
\end{proof}
\hide{
\begin{proof}[of Theorem~\ref{wal-e1q-PoA}]

By Lemma~\ref{k-bad probability}, the probability that $\vec n$ is $(k+1,  \frac{\max_s\{v_i(s)\}}{2^c}, \max_s\{v_i(s)\})$-bad or $n_j \leq k+1$ for some $j$ is less than 
\begin{eqnarray}
&&\sum_{1 \leq j \leq l} F(j,N) \left[ \frac{1}{\frac{\max_s\{v_i(s)\}}{2^c}}\max_s\{v_i(s)\}\Lambda(l, k+1) + k+2\right] \nonumber\\
&&~~~~~~ = \sum_{1 \leq j \leq l} F(j,N) \left[ 2^c \Lambda(l, k+1) + k+2\right]. \nonumber
\end{eqnarray} 
So, for any integer $c' $,
\begin{align*}
\mathbb{E}_{\vec n}[u_i(v_i, b_{-i}) ] &\geq \mathbb{E}_{\vec n}\Bigg[v_i(x_i(v_i, v_{-i})) - \sum_{s \in x_i(v_i, v_{-i})} p_s(\vec n; (b_i, b_{-i})) \\
&~~~~~~~~~ - \sum_{c = 1}^{c'} \mathds{1}\Big[\text{$\vec n$ is $(k+1,\frac{\max_s\{v_i(s)\}}{2^c},\max_s\{v_i(s)\})$-bad and} \nonumber\\
&~~~~~~~~~~~~~~~~~~~~~~~~~~~~~~~~~~~ \text{$(k+1,\frac{\max_s\{v_i(s)\}}{2^{c-1}}, \max_s\{v_i(s)\})$-good}\Big] \\
&~~~~~~~~~~~~~~~~~~\cdot  \Big| x_i(v_i, v_{-i})  \cap d_i \Big| \cdot (k+1) \frac{\max_s\{v_i(s)\}}{2^{c-1}}   \\
&~~~~~~~~~ - \mathds{1}\Big[\text{$\vec n$ is $(k+1,\frac{\max_s\{v_i(s)\}}{2^{c'-1}}, \max_s\{v_i(s)\})$-good}\Big] \nonumber\\
&~~~~~~~~~~~~~~~~~~~~~~~~~~~~~~~~~~~ \cdot  \Big| x_i(v_i, v_{-i})  \cap d_i \Big| \cdot (k+1) \frac{\max_s\{v_i(s)\}}{2^{c'}}  \Bigg] \\
&\geq \mathbb{E}_{\vec n}\Bigg[v_i(x_i(v_i, v_{-i})) - \sum_{s \in x_i(v_i, v_{-i})} p_s(\vec n; (b_i, b_{-i})) \Bigg] \\
&~~~~~~~~~ - \sum_{c = 1}^{c'} \cdot \sum_{1 \leq j \leq l} F(j,N) \left[2^c \Lambda(l, k+1) + k+2\right] \cdot  k \cdot (k+1) \frac{\max_s\{v_i(s)\}}{2^{c-1}} \\
&~~~~~~~~~ - k \cdot (k+1) \frac{\max_s\{v_i(s)\}}{2^{c'}}  \\
&\geq \mathbb{E}_{\vec n}\Bigg[v_i(x_i(v_i, v_{-i})) - \sum_{s \in x_i(v_i, v_{-i})} p_s(\vec n; (b_i, b_{-i})) \\
&~~~~~~~~~ - c' \cdot \sum_{1 \leq j \leq l} F(j,N) \left[ 2 \Lambda(l, k+1) + k+2\right] \cdot  k \cdot (k+1) \cdot \max_s\{v_i(s)\}    \\
&~~~~~~~~~ - k \cdot (k+1) \frac{\max_s\{v_i(s)\}}{2^{c'}} \Bigg]. \\
\end{align*}
Summing over all the bidders and integrating w.r.t. $\vec v$ and $\vec b$ gives
\begin{align*}
\sum_i \mathbb{E}_{\vec v, \vec b, \vec n}[u_i(v_i, b_{-i})] &\geq  \sw(\opt) - \text{R}(b_i, b_{-i}) \\
&~~~~~~ - N \cdot c' \cdot \sum_{1 \leq j \leq l} F(j,N) \left[ 2 \Lambda(l, k+1) + k+2\right] \cdot k \cdot (k + 1) \cdot \zeta \cdot l \\
&~~~~~~ - N \cdot k \cdot (k + 1) \frac{1}{2^{c'}} \cdot \zeta \cdot l.
\end{align*}

Using the smooth technique for Bayesian settings~\cite{Syrgkanis:2013:CEM:2488608.2488635} yields
\begin{eqnarray}
\text{SW(NE)} &\geq&  \Bigg(1 - \frac{\zeta \cdot l \cdot k \cdot (k + 1) \frac{1}{2^{c'}}}{\rho} \nonumber \\
&&~~~~~~ - \frac{\zeta \cdot l \cdot  c' \cdot \sum_{1 \leq j \leq l} F(j,N) \left[ 2 \Lambda(l, k+1) + k+2\right] \cdot k \cdot (k + 1)}{\rho}\Bigg)  \sw(\opt). \nonumber
\end{eqnarray}

Let $Y = \sum_{1 \leq j \leq l} F(j,N) \left[ 2 \Lambda(l, k+1) \right]$. Set $c' = \lceil \log_2 \frac{1}{Y} - \log_2 \log_2 \frac{1}{Y} \rceil$; then $\frac{1}{2^{c'}} \leq Y \log_2  \frac{1}{Y}$. So, 
\begin{eqnarray}
\text{SW(NE)} \geq \Bigg(1 - \frac{3 \cdot k \cdot (k+1) \cdot \zeta \cdot l}{\rho} \cdot  Y \cdot  \lceil \log_2 \frac{1}{Y}\rceil  \Bigg) \sw(\opt). \nonumber
\end{eqnarray}
\end{proof}

}
\hide{
We note the following corollary to Theorem~\ref{wal-e1q-PoA}.
\begin{corollary}
\label{col:regret-wal}
In a large Walrasian auction which satisfies Assumptions~\ref{ass:bdd-val} and~\ref{ass:sw-grows}, if $v_i$ and $b_i$  satisfy the gross substitutes property for all $i$, then
% the Bounded Valuation and Market Size properties,
% If $\text{SW(OPT)} > \rho N$, then
\[
\sum_i \mathbb{E}_{\bn, \bv, \bb} [u_i(v_i, b_{-i})] \geq \Bigg(1 - \frac{3 \cdot k \cdot (k+1) \cdot \zeta \cdot m}{\rho} \cdot  Y \cdot  \lceil \log_2 \frac{1}{Y}\rceil \Bigg) \sw(\opt) - \mathbb{E}_{\bb}[\text{R}(b_i, b_{-i})]   
\]
where $Y = \frac{m}{L} \left[ 2 m {{k+1+m}\choose{m}} \right]$.
\end{corollary}
}

\paragraph{Comparison of our methodology with that of~\cite{FeldmanILRS:15}}
We will be considering the combinatorial auction in \cite{FeldmanILRS:15} which uses separate
auctions for each type of good, more specifically a $(c+1)$-st price auction when there are $c$ copies of the good.
To facilitate a comparison, we adjust their notation to match the notation we have been using and 
reduce its generality\footnote{In fact, the comparison applies in full generality.}.
\hide{Their methodology is more structured than ours and to obtain a bound it requires} They begin by defining a notion of smooth in the large which in the current context amounts to showing
\begin{eqnarray}\label{feldman-eq}
\sum_i U(v_i,b_{-i}) \ge (1 - \epsilon) \sw(\opt) - R(\bb).
\end{eqnarray}
To obtain such bounds, they propose the following methodology: It entails identifying
an approximate utility function $U(v_i,b_{-i})$ and then showing the following two bounds:
\begin{itemize}
\item The approximate and actual utilities are close:
For all $\bb$, $|u_i(\bb) - U_i(\bb) | \le \eps$.
\item The standard smoothness formulation applies to the approximate utility:
For all $i, \bv, \bb$, $\sum_i U(v_i,b_{-i}) \ge \sw(\opt) - R(\bb)$.
\end{itemize}
One can then deduce that
$\sum_i U(v_i,b_{-i}) \ge \sw(\opt) - R(\bb) -N\eps$,
which, on taking expectations,
is exactly the bound we obtain for our auction. 
With the assumption that $\sw(\opt) = \rho N$ one obtains \eqref{feldman-eq}.
However, it not clear that we can specify an approximate utility $U$ as specified in the framework
of~\cite{FeldmanILRS:15}.
In particular, handling the expected bound on valuations in this framework, 
rather than the fixed bound used by Feldman et al.,
appears challenging.

\section{Large Fisher Markets}
\label{sec:large-fisher}

\hide{
Recall that Eisenberg Gale markets are those economies for which the equilibria are exactly the solutions to the following convex program called the Eisenberg-Gale convex program:
\begin{align*}
\max_x  ~~\sum_{i=1}^{n} & e_i \cdot \log(u_i(x_{i1}, x_{i2}, \cdots x_{im})) \nonumber\\
&\text{s.t.} ~~~~ \forall j:~~~~ \sum_i x_{ij} \leq 1 \nonumber\\
& ~~~~~~~~~\forall i, j:~~ x_{ij} \geq 0. \nonumber
\end{align*}

\begin{fact}
The utility functions in Eisenberg Gale programs are homogeneous of degree 1.
\end{fact}

In order to have a meaningful social welfare measure, we put the various bidders' utilities on a common scale as follows.

\begin{definition}
The bidders' utilities are consistently scaled if there is a parameter $t > 0$ such that for every bidder $i$, $u_i(v_i, v_{-i}) = t e_i$. \footnote{WLOG, we may assume that $t = 1$.}
\end{definition}
\begin{theorem}
\label{fisher-market-poa}
If the bidders' utilities are consistently scaled then
\begin{align*}
\sw(\ne) \geq e^{-m \cdot \frac{\max_i{e_i}}{\sum_i e_i}} \cdot \sw(\opt).
\end{align*}
\end{theorem}
}

Theorem~\ref{thm:FisherPoA},
which states that the PoA of an $m$-good market of largeness $L$ is at most $e^{m/L}$,
will follow from the next lemma.

\begin{lemma}
\label{fisher-market-multi-poa}
For any bidding profile $\mathbf{b}$ and any value profile $\mathbf{v}$ which are homogeneous of degree 1, concave, continuous, monotone and which satisfy the  gross substitutes property,
\begin{align*}
\sum_{i = 1}^{n}  e_i \cdot \log(u_i(v_i, b_{-i})) \geq \sum_{i = 1}^n  e_i \cdot \log(u_i(v_i, v_{-i})) - m \cdot \max_i{e_i}. \nonumber
\end{align*}
\end{lemma}

\pfof{Theorem~\ref{thm:FisherPoA}} We proof our results of GPoA and PoA separately.

\begin{itemize}
\item {\bf PoA bound} On exponentiating the expressions on both sides in the statement of Lemma~\ref{fisher-market-multi-poa}, we obtain
\begin{eqnarray}
\prod_i u_i(v_i, b_{-i})^{e_i} \geq \frac{1}{e^{m \cdot \max_i e_i}} \prod_i u_i(v_i, v_{-i})^{e_i}.  \nonumber
\end{eqnarray}
Therefore, $\prod_i \Big( \frac{u_i(v_i, b_{-i})}{u_i(v_i, v_{-i})} \Big)^{e_i} \geq  \frac{1}{e^{m \cdot \max_i e_i}}$. Using the weighted GM-AM inequality, we obtain
\begin{eqnarray}
\frac{\sum_i e_i  \frac{u_i(v_i, b_{-i})}{u_i(v_i, v_{-i})} }{ \sum_i e_i  } \geq \Bigg(\prod_i \Big( \frac{u_i(v_i, b_{-i})}{u_i(v_i, v_{-i})} \Big)^{e_i}\Bigg)^{\frac{1}{\sum_i e_i}} \geq \Bigg(\frac{1}{e^{m \max_i e_i }}\Bigg)^{\frac{1}{\sum_i e_i}} = e^{-\frac{m \max_i e_i}{\sum_i e_i}}. \nonumber
\end{eqnarray}
Since  for all $i$, $u_i(v_i, v_{-i}) = t e_i$, $\sum_i u_i(v_i, b_{-i}) \geq e^{-\frac{m \max_i e_i}{\sum_i e_i}} \sum_i u_i(v_i, v_{-i})$. The theorem follows on applying the smooth technique.
\item {\bf GPoA bound} 
According to Lemma~\ref{fisher-market-multi-poa},
\begin{align*}
\prod_i u_i(v_i, b_{-i})^{e_i} \geq \frac{1}{e^{m \cdot \max_i e_i}} \prod_i u_i(v_i, v_{-i})^{e_i}.
\end{align*}
Therefore,
\begin{align*}
\prod_i \mathbb{E}_{\mathbf{b}}[u_i(v_i, b_{-i})]^{e_i} &\geq \prod_i e^{\mathbb{E}_{\mathbf{b}}[\log u_i(v_i, b_{-i})] e_i} = e^{\sum_i \mathbb{E}_{\mathbf{b}}[\log u_i(v_i, b_{-i})] e_i} \\
&= e^{\mathbb{E}_{\mathbf{b}}[\sum_i \log (u_i(v_i, b_{-i})^{e_i})]} = e^{\mathbb{E}_{\mathbf{b}}[\log (\prod_i u_i(v_i, b_{-i})^{e_i})]} \\
&\geq e^{\mathbb{E}_{\mathbf{b}}[\log  \frac{1}{e^{m \cdot \max_i e_i}} \prod_i u_i(v_i, v_{-i})^{e_i}]} = \frac{1}{e^{m \cdot \max_i e_i}} \prod_i u_i(v_i, v_{-i})^{e_i}.
\end{align*}

By applying the Nash equilibrium condition, $\mathbb{E}_{\mathbf{b}}[u_i(b_i, b_{-i})] \geq \mathbb{E}_{\mathbf{b}}[u_i(v_i, b_{-i})]$, the GPoA bound follows.
\end{itemize}
\end{proof}

To prove Lemma~\ref{fisher-market-multi-poa}, we need the following claim;
intuitively, it states that a single bidder can cause the prices to decrease by only a small amount.

\begin{lemma}
\label{lem::price}
$\bp(v_i, b_{-i}) \leq \bp(b_i, b_{-i}) + \max_i e_i \cdot \mathsf{1}$.
\end{lemma}

\pfof{Lemma~\ref{fisher-market-multi-poa}}
Consider the dual of the Eishenberg-Gale convex program:
\begin{align*}
\min_p \max_x  \sum_{i=1}^{n} & e_i \cdot \log(u_i(x_{i1}, x_{i2}, \cdots, x_{im})) - \sum_{i,j} p_j x_{ij} + \sum_j p_j \nonumber\\
&\text{s.t.} ~~~~ \forall j:~~ p_j \ge 0 
 ~~~~\text{and}~~~~\forall i, j:~~ x_{ij} \geq 0. \nonumber
% &\text{s.t.} ~~~~ \forall j:~~~~ p_j \ge 0 \nonumber\\
% & ~~~~~~~~~\forall i, j:~~ x_{ij} \geq 0. \nonumber
\end{align*}

Let $\bp$ denote an arbitrary collection of prices, and $\bp^*$ denote the prices with truthful bids. 
Since $\bp^*$ minimizes the dual program,
\begin{align}
\label{eqn:opt-term}
&\max_{\bx \geq 0}  ~~~\sum_{i=1}^{n} e_i \cdot \log(u_i(x_{i1}, x_{i2}, \cdots, x_{im})) - \sum_{i,j} p_j x_{ij} + \sum_j p_j \\
&~~~~~~~~~~~~~~~\geq \max_{\bx \geq 0}  ~~~\sum_{i=1}^{n} e_i \cdot \log(u_i(x_{i1}, x_{i2}, \cdots, x_{im})) - \sum_{i,j} p^*_j x_{ij} + \sum_j p^*_j. \nonumber
\end{align}

Let $\tilde{x}_{ij}$ be an allocation over all goods $j$ and bidders $i$ at prices $p$
that maximize~\eqref{eqn:opt-term}.  
As $u_i$ is homogeneous of degree 1, 
$u_i$ is differentiable in the direction $\bxi$. 
It follows that 
\begin{align}
\label{eqn:diff-cond}
\lim_{\epsilon \rightarrow 0} \frac{[e_i \cdot \log u_i((1 + \epsilon) \tilde{\bx}_{i}) - \sum_{j} p_j (1 + \epsilon) \tilde{x}_{ij}] - [e_i \cdot \log u_i(\tilde{\bx}_{i}) - \sum_{j} p_j \tilde{x}_{ij}]}{\epsilon} = 0.
\end{align}

The LHS of \eqref{eqn:diff-cond} equals $e_i - \sum_j p_j \tilde{x}_{ij}$, implying that
$e_i = \sum_j p_j \tilde{x}_{ij}$.
Therefore,
\begin{eqnarray}
&\max_{\bx : \forall i \sum x_{ij} p_j = e_i}  ~~~\sum_{i=1}^{n} e_i \cdot \log(u_i(x_{i1}, x_{i2}, \cdots x_{im})) + \sum_j p_j \nonumber\\
&~~~~~~~~~\geq \max_{\bx : \forall i \sum x_{ij} p^*_j = e_i}  ~~~\sum_{i=1}^{n} e_i \cdot \log(u_i(x_{i1}, x_{i2}, \cdots x_{im})) + \sum_j p^*_j. \label{multi-poa-ineq}
\end{eqnarray}

If all the prices stay the same or increase, a buyer's optimal utility stays the same or reduces.
Using the price upper bound from Lemma~\ref{lem::price}, it follows that
\begin{align*}
\sum_{i = 1}^{n}  e_i \cdot \log(u_i(v_i, b_{-i})) 
&\geq \sum_{i = 1}^n \max_{\bx : \forall i \sum x_{ij} (p_j(b_i, b_{-i})  + \max_{i'} e_{i'}) = e_i} e_i \cdot \log(u_i(x_{i1}, x_{i2}, \cdots, x_{im})) \nonumber  \\
&=\sum_{i = 1}^n \max_{\bx : \forall i \sum x_{ij} (p_j(b_i, b_{-i}) + \max_{i'} e_{i'}) = e_i} e_i \cdot \log(u_i(x_{i1}, x_{i2}, \cdots, x_{im})) \nonumber\\
&~~~~~~+ \sum_j (p_j(b_i, b_{-i}) + \max_{i'} e_{i'}) - \sum_j (p_j(b_i, b_{-i}) + \max_{i'} e_{i'})\nonumber \\
~~~~~~~~~~~~&\geq  \sum_{i = 1}^n \max_{\bx : \forall i \sum x_{ij} p_j^* = e_i} e_i \cdot \log(u_i(x_{i1}, x_{i2}, \cdots, x_{im})) \nonumber \\
&~~~~~~+ \sum_j p^*_j - \sum_j (p_j(b_i, b_{-i}) + \max_{i'} e_{i'})~~~~~~~~~~~~\mbox{by \eqref{multi-poa-ineq}} \nonumber\\&\ge  \sum_{i = 1}^n \max_{\bx : \forall i \sum_j x_{ij} p_j^* = e_i} e_i \cdot \log(u_i(x_{i1}, x_{i2}, \cdots, x_{im})) - m \max_i e_i \nonumber\\
&~~~~~~~~~~~~~~~~~~~~~~~~~~~~~~~\text{as}~ \sum_j p_j^* = \sum_i e_i = \sum_j p_j(b_i,b_{-i}) \\
&=  \sum_{i = 1}^n e_i \log(u_i(v_i, v_{-i})) - m \max_i e_i.  \nonumber
\end{align*}
\end{proof}

The proof of Lemma~\ref{lem::price} uses the following notation 
and follows from Lemmas~\ref{lem::prices::les}
and~\ref{lem::prices::geq} below.
$\bp$ denotes the prices when the $i$th bidder is not participating and the bidding profile is $b_{-i}$; 
$\bx$ denotes the resulting allocation. 
Similarly, $\hat{\bp}$ denotes the prices when the bidding profile is $(b_i, b_{-i})$;
$\hat{\bx}$ denotes the resulting allocation. 

\begin{lemma}\label{lem::prices::les}
$ \bp \preceq \hat{\bp} = \bp(b_i, b_{-i}).$
\end{lemma}

\begin{lemma}\label{lem::prices::geq}
$\hat{\bp} \preceq  \bp + e_i \cdot \mathsf{1}. $
\end{lemma}

\pfof{Lemma~\ref{lem::price}}
Lemmas~\ref{lem::prices::les} and \ref{lem::prices::geq} also apply to prices $\bp(v_i, b_{-i})$ as well as to $\hat{\bp}$.
So $\bp(v_i,b_{-i}) \le \bp + e_i \cdot \mathsf{1} \le \bp(b_i,b_{-i}) + e_i \cdot \mathsf{1} \leq \bp(b_i,b_{-i}) + \max_i e_i \cdot \mathsf{1}$.
\end{proof}

Lemma \ref{lem::prices::geq} follows readily from Lemma \ref{lem::prices::les}.
\pfof{Lemma~\ref{lem::prices::geq}}
Since $\mathsf{1} \cdot \bp + e_i = \mathsf{1} \cdot \hat{\bp}$ and $\bp \preceq \hat{\bp}$, the lemma follows.
\end{proof}

We finish by proving that Lemma~\ref{lem::prices::les} holds in two scenarios: single-demand WGS utility functions and linear utility functions.

\subsection{Single-Demand WGS Utility Functions}

\pfof{Lemma~\ref{lem::prices::les}}
For a contradiction, we suppose there is an item $j$ such that $p_j > \hat{p}_j$.

Let $\epsilon > 0$ be a very small constant such that $\epsilon < p_{k}$ for all $p_{k} \neq 0$ and $\epsilon < \hat{p}_{k}$ for all $\hat{p}_{k} \neq 0$. 

Let $\bp'$ denote the following collection of prices:
$p'_{k} = p_{k}$ if $p_{k} \neq 0$, and $p'_{k} = \epsilon$ otherwise.
We consider the resulting demands for a bidder $h \neq i$. 
Recall that $\bx_{h}$ denotes bidder $h$'s demand at prices $\bp$.
$\bx'_{h}$ will denote her demand at prices $\bp'$.
By the WGS property, $x'_{hk} = x_{hk}$ if $p_{k} \neq 0$, and $x'_{hk} = 0$ 
if $p_{k} = 0$, i.e.\ if $p'_k = \epsilon$.\footnote{Changing the prices from $\bp$ to $\bp'$,
one by one,
by setting $p'_k$ to $\eps$, which happens when $p_k = 0$,
only increases the demand for other goods, but as no spending is released by this
price increase, these demands are in fact unchanged.}

Analogously, let $\hat p'_{k} = \hat p_{k}$ if $\hat p_{k} \neq 0$, and $\hat p'_{k} = \epsilon$ otherwise.
Let $\hat \bx_{h}$ denote bidder $h$'s demand at prices $\hat{\bp}$, and 
$\hat \bx'_{h}$  her demand at prices $\hat{\bp}'$.
Again, $\hat x'_{hk} = \hat x_{hk}$ if $\hat p_{k} \neq 0$, and $\hat x'_{hk} = 0$ 
if $\hat p'_{k} = \epsilon$.

Now, we look at those items $l$ which have the smallest ratio between $p'_l$ and $\hat{p}'_l$.
\[
S = \Bigg\{l~\Bigg|~ \frac{\hat{p}'_{l} } {p'_{l} } = \min_{k}\frac{\hat{p}'_{k} } {p'_{k} }  \Bigg\}.
\]

By assumption, $p_j > \hat{p}_j$; therefore $p'_j > \hat{p}'_j$. 
Thus, for $l \in S$, $\frac{\hat{p}'_{l} } {p'_{l} } < 1$. 
For simplicity, let $\eta$ denote this ratio.
Note that this inequality implies $p'_l > \eps$, and thus $p_l = p'_l >0$.
Also,
\begin{equation}
\label{eqn:p-pp-relation}
p_l = p'_l > \hat p'_l > 0.
\end{equation}

We now consider the following procedure:

First multiply $\bp'$ by $\eta$. 
By the homogeneity of the utility function, 
bidder $h$'s demand at prices $\eta \cdot \bp'$ will be $\frac{1}{\eta} \bx'_{h}$. 
Note that $\eta \cdot p'_{l} = \hat{p}'_{l}$ for any $l \in S$ 
and $\eta \cdot p'_{k} < \hat{p}'_{k}$ for any $k \notin S$.

Second, increase the prices of $\eta \cdot \bp'$ to $\hat{\bp}'$. 
Since for $l \in S$ the two prices are the same, 
by the Gross Substitutes property, $\hat{x}'_{hl} \geq \frac{1}{\eta} x'_{hl}$ for any $l \in S$.
%Since $\eta < 1$, and by~\eqref{eqn:p-pp-relation},
%$p_{l} > \hat{p}_{l} \geq 0$ for any $l \in S$, $\hat{x}'_{hk} > x'_{hk}$. 

Summing over all the bidders except $i$, 
\[\sum_{h \neq i} \hat{x}'_{hl} \geq  \frac{1}{\eta} \sum_{h \neq i} x'_{hl} ~~~~~~\mbox{for $l\in S$}.\]

 By~\eqref{eqn:p-pp-relation},
$p_{l} > 0$ for any $l \in S$; hence $\sum_{h \neq i} x'_{hl} = \sum_{h \neq i} x_{hl} = 1$. So,  since $\eta < 1$,
\begin{eqnarray}
\label{eqn::fish::price}
\sum_{h \neq i} \hat{x}'_{hl} > \sum_{h \neq i} x'_{hl} =  \sum_{h \neq i} x_{hl} = 1~~~~~~\mbox{for $l\in S$}.
\end{eqnarray}

For all $h$ and $l$, $\hat{x}_{hl} \geq \hat{x}'_{hl}$. Therefore, 
\begin{align*}
\sum_{h} \hat{x}_{hl} &\geq \sum_{h \neq i} \hat{x}'_{hl} 
> \sum_{h \neq i} x_{hl} = 1 ~~~~~~\mbox{for $l \in S$}.
\end{align*}
As $\sum_{h} \hat{x}_{hl} \leq 1$, this is impossible and yields a contradiction.
\end{proof}

\subsection{Linear Utility function}
\pfof{Lemma~\ref{lem::prices::les}}
For a contradiction, we suppose there is an item $j$ such that $p_j > \hat{p}_j$.
Now, we look at those items $j$ which have the smallest ratio between $p_l$ and $\hat{p}_l$.
\[
S = \Bigg\{l ~\Bigg|~ \frac{\hat{p}_{l} } {p_{l} } = \min_{k}\frac{\hat{p}_{k} } {p_{k} }  \Bigg\}.
\]

For simplicity, we set $\frac{0}{x} = 0$ for $x > 0$, $\frac{0}{0} = 1$ and $\frac{x}{0} = +\infty$ for $x > 0$.

For linear utility functions, we use the following observation: 
if at prices $\bp$ a bidder's favorite items include some items in $S$, 
then at prices $\hat{\bp}$ her favorite items will all be in $S$. 

Therefore, as the price of each good equals the total spending on that good,
\[
\sum_{l \in S} p_l \leq \sum_{l \in S} \hat p'_l.
\]

This implies that $\min_{k}\frac{\hat{p}_{k}}{p_{k}} = 1$, and the lemma follows.
\end{proof}

\section{Acknowledgments}
The authors would like to thank the anonymous referees of
this paper and an earlier version for their thoughtful comments.
This work was partly supported by NSF awards
CCF-1217989 and CCF-1527568.
The work was partly performed while the first author was attending the fall 2015 program in Economics and Computation at the Simons Institute for Theoretical Computer Science at UC Berkeley.

% Acknowledgments

\bibliographystyle{IEEEtran}
\bibliography{references,more-refs}

\begin{thebibliography}{10}
\providecommand{\url}[1]{#1}
\csname url@samestyle\endcsname
\providecommand{\newblock}{\relax}
\providecommand{\bibinfo}[2]{#2}
\providecommand{\BIBentrySTDinterwordspacing}{\spaceskip=0pt\relax}
\providecommand{\BIBentryALTinterwordstretchfactor}{4}
\providecommand{\BIBentryALTinterwordspacing}{\spaceskip=\fontdimen2\font plus
\BIBentryALTinterwordstretchfactor\fontdimen3\font minus
  \fontdimen4\font\relax}
\providecommand{\BIBforeignlanguage}[2]{{%
\expandafter\ifx\csname l@#1\endcsname\relax
\typeout{** WARNING: IEEEtran.bst: No hyphenation pattern has been}%
\typeout{** loaded for the language `#1'. Using the pattern for}%
\typeout{** the default language instead.}%
\else
\language=\csname l@#1\endcsname
\fi
#2}}
\providecommand{\BIBdecl}{\relax}
\BIBdecl

\bibitem{Roughgarden2015}
T.~Roughgarden, ``Intrinsic robustness of the price of anarchy,'' \emph{J.
  ACM}, vol.~62, no.~5, pp. 32:1--32:42, Nov. 2015.

\bibitem{Roughgarden2012}
------, ``The price of anarchy in games of incomplete information,'' in
  \emph{Proceedings of the Thirteenth ACM Conference on Electronic Commerce},
  ser. EC '12.\hskip 1em plus 0.5em minus 0.4em\relax New York, NY, USA: ACM,
  2012, pp. 862--879.

\bibitem{Syrgkanis:2013:CEM:2488608.2488635}
V.~Syrgkanis and E.~Tardos, ``Composable and efficient mechanisms,'' in
  \emph{Proceedings of the Forty Fifth Annual ACM Symposium on Theory of
  Computing}, ser. STOC '13.\hskip 1em plus 0.5em minus 0.4em\relax New York,
  NY, USA: ACM, 2013, pp. 211--220.

\bibitem{FeldmanILRS:15}
M.~Feldman, N.~Immorlica, B.~Lucier, T.~Roughgarden, and V.~Syrgkanis, ``The
  price of anarchy in large games,'' in \emph{Proceedings of the Forty Eighth
  Annual ACM Symposium on Theory of Computing}, ser. STOC '16, 2016.

\bibitem{RoughgardenT2002}
T.~Roughgarden and E.~Tardos, ``How bad is selfish routing?'' \emph{J. ACM},
  vol.~49, no.~2, pp. 236--259, Mar. 2002.

\bibitem{syrgkanis2012bayesian}
V.~Syrgkanis, ``Bayesian games and the smoothness framework,'' \emph{arXiv
  preprint arXiv:1203.5155}, 2012.

\bibitem{GulStac99}
F.~Gul and E.~Stacchetti, ``Walrasian equilibrium with gross substitutes,''
  \emph{Journal of Economic Theory}, vol.~87, pp. 95--124, 1999.

\bibitem{BabaioffLNP14}
M.~Babaioff, B.~Lucier, N.~Nisan, and R.~Paes~Leme, ``On the efficiency of the
  walrasian mechanism,'' in \emph{Proceedings of the Fifteenth ACM Conference
  on Economics and Computation}, ser. EC '14.\hskip 1em plus 0.5em minus
  0.4em\relax New York, NY, USA: ACM, 2014, pp. 783--800.

\bibitem{Swinkels2001}
J.~M. Swinkels, ``Efficiency of large private value auctions,''
  \emph{Econometrica}, vol.~69, no.~1, pp. 37--68, 2001.

\bibitem{BhawalkarR2011}
K.~Bhawalkar and T.~Roughgarden, ``Welfare guarantees for combinatorial
  auctions with item bidding,'' in \emph{Proceedings of the Twenty Second
  Annual ACM-SIAM Symposium on Discrete Algorithms}, ser. SODA '11.\hskip 1em
  plus 0.5em minus 0.4em\relax SIAM, 2011, pp. 700--709.

\bibitem{ChristodoulouKS2008}
G.~Christodoulou, A.~Kov\'{a}cs, and M.~Schapira, ``Bayesian combinatorial
  auctions,'' \emph{J. ACM}, vol.~63, no.~2, pp. 11:1--11:19, Apr. 2016.

\bibitem{FuKL2012}
H.~Fu, R.~Kleinberg, and R.~Lavi, ``Conditional equilibrium outcomes via
  ascending price processes with applications to combinatorial auctions with
  item bidding,'' in \emph{Proceedings of the Thirteenth ACM Conference on
  Electronic Commerce}, ser. EC '12.\hskip 1em plus 0.5em minus 0.4em\relax New
  York, NY, USA: ACM, 2012, pp. 586--586.

\bibitem{HyllandZec1979}
A.~Hylland and R.~Zeckhauser, ``{The Efficient Allocation of Individuals to
  Positions},'' \emph{Journal of Political Economy}, vol.~87, no.~2, pp.
  293--314, April 1979.

\bibitem{Varian1974}
H.~Varian, ``Equity, envy, and efficiency,'' \emph{Journal of Economic Theory},
  vol.~9, no.~1, pp. 63--91, 1974.

\bibitem{Myerson2000}
R.~B. Myerson, ``Large {P}oisson games,'' \emph{Journal of Economic Theory},
  vol.~94, no.~1, pp. 7--45, 2000.

\bibitem{AzevBudish2012}
E.~M. Azevedo and E.~Budish, ``Strategyproofness in the large as a desideratum
  for market design,'' in \emph{Proceedings of the Thirteenth ACM Conference on
  Electronic Commerce}, ser. EC '12.\hskip 1em plus 0.5em minus 0.4em\relax New
  York, NY, USA: ACM, 2012, pp. 55--55.

\bibitem{hsu2015prices}
J.~Hsu, J.~Morgenstern, R.~Rogers, A.~Roth, and R.~Vohra, ``Do prices
  coordinate markets?'' in \emph{Proceedings of the Forty Eighth Annual ACM
  Symposium on Theory of Computing}, ser. STOC '16, 2016.

\bibitem{RobertsPosle1976}
D.~J. Roberts and A.~Postlewaite, ``{The Incentives for Price-Taking Behavior
  in Large Exchange Economies},'' \emph{Econometrica}, vol.~44, no.~1, pp.
  115--27, January 1976.

\bibitem{JacksonMan97}
M.~O. Jackson and A.~M. Manelli, ``Approximately competitive equilibria in
  large finite economies,'' \emph{J. Economic Theory}, vol.~77, no.~3, pp.
  354--376, 1997.

\bibitem{AlNajjar:2007}
N.~I. Al-Najjar and R.~Smorodinsky, ``The efficiency of competitive mechanisms
  under private information,'' \emph{Journal of Economic Theory}, vol. 137,
  no.~1, pp. 383--403, 2007.

\bibitem{BranzeiCDFFZ14}
S.~Br{\^a}nzei, Y.~Chen, X.~Deng, A.~Filos-Ratsikas, S.~K.~S. Frederiksen, and
  J.~Zhang, ``The fisher market game: Equilibrium and welfare,'' in
  \emph{Twenty Eighth AAAI Conference on Artificial Intelligence}, 2014, pp.
  587--593.

\bibitem{ChenDZ2011}
N.~Chen, X.~Deng, and J.~Zhang, ``\BIBforeignlanguage{English}{How profitable
  are strategic behaviors in a market?}'' in
  \emph{\BIBforeignlanguage{English}{Algorithms – ESA 2011}}, ser. Lecture
  Notes in Computer Science, C.~Demetrescu and M.~M. Halld{\'o}rsson,
  Eds.\hskip 1em plus 0.5em minus 0.4em\relax Springer Berlin Heidelberg, 2011,
  vol. 6942, pp. 106--118.

\bibitem{ChenDZZ2012}
N.~Chen, X.~Deng, H.~Zhang, and J.~Zhang,
  ``\BIBforeignlanguage{English}{Incentive ratios of fisher markets},'' in
  \emph{\BIBforeignlanguage{English}{Automata, Languages, and Programming}},
  ser. Lecture Notes in Computer Science, A.~Czumaj, K.~Mehlhorn, A.~Pitts, and
  R.~Wattenhofer, Eds.\hskip 1em plus 0.5em minus 0.4em\relax Springer Berlin
  Heidelberg, 2012, vol. 7392, pp. 464--475.

\bibitem{Kalai2004}
E.~Kalai, ``Large robust games,'' \emph{Econometrica}, vol.~72, no.~6, pp.
  1631--1665, 2004.

\bibitem{KalaiShmaya2013}
E.~Kalai and E.~Shmaya, ``Large repeated games with uncertain fundamentals {I}:
  Compressed equilibrium,'' 2013.

\bibitem{PaiRU2014}
M.~M. Pai, A.~Roth, and J.~Ullman, ``An anti-folk theorem for large repeated
  games with imperfect monitoring,'' \emph{CoRR}, vol. abs/1402.2801, 2014.

\bibitem{GradwohlReingold2008}
R.~Gradwohl and O.~Reingold, ``Fault tolerance in large games,'' in
  \emph{Proceedings of the Ninth ACM Conference on Electronic Commerce}, ser.
  EC '08.\hskip 1em plus 0.5em minus 0.4em\relax New York, NY, USA: ACM, 2008,
  pp. 274--283.

\bibitem{ChenTeng2009}
X.~Chen and S.-H. Teng, ``Spending is not easier than trading: on the
  computational equivalence of fisher and arrow-debreu equilibria,'' in
  \emph{Algorithms and Computation}.\hskip 1em plus 0.5em minus 0.4em\relax
  Springer, 2009, pp. 647--656.

\bibitem{VaziraniYann2011}
V.~V. Vazirani and M.~Yannakakis, ``Market equilibrium under separable,
  piecewise-linear, concave utilities,'' \emph{J. ACM}, vol.~58, no.~3, pp.
  10:1--10:25, Jun. 2011.

\bibitem{JainVaz2007}
K.~Jain and V.~V. Vazirani, ``Eisenberg-gale markets: Algorithms and structural
  properties,'' in \emph{Proceedings of the Thirty Ninth Annual ACM Symposium
  on Theory of Computing}, ser. STOC '07.\hskip 1em plus 0.5em minus
  0.4em\relax New York, NY, USA: ACM, 2007, pp. 364--373.

\bibitem{EisenbergGale1959}
E.~Eisenberg and D.~Gale, ``Consensus of subjective probabilities: The
  pari-mutuel method,'' \emph{Ann. Math. Statist.}, vol.~30, no.~1, pp.
  165--168, 03 1959.

\end{thebibliography}
% History dates

% Electronic Appendix

\newpage
 Appendix
\appendix
\section{Omitted Proofs}
\label{sec:appendix}

\subsection{Proofs from Section~\ref{sec:results}}

\pfof{Lemma~\ref{lem:high-sw}}
Let $\#\text{items}_j$ denote the number of copies of good $j$ that are present,
and let $N_j$ denote the number of buyers for which good $j$ has the largest
expected value (breaking ties arbitrarily).
% Given value profile $v$, there is at least one good for which at least $\frac{N}{l}$ bidders have non-zero value. Let $j$ denote this good.
By Chebyshev's Theorem, $\Pr\big[\#\text{items}_j > \mathbb{E}[\#\text{items}_j] - t \cdot \Gamma(\#\text{items}_j)\big] \geq 1 - \frac{1}{t^2}$.
We set $t$ equal $1 + \lambda$, where $\lambda$ is the parameter in Assumption~\ref{ass:mark-size}.
Then by Assumption~\ref{ass:mark-size},
$\Pr\Big[\#\text{items}_j > \lambda^2 \cdot \mathbb{E}[\#\text{items}_j]\Big] \geq \frac{2 \lambda + \lambda^2}{(1 + \lambda)^2}$,
which implies $\Pr[\#\text{items}_j >\lambda^2 \alpha N] \geq \frac{2 \lambda + \lambda^2}{(1 + \lambda)^2}$.
If at least $\lambda^2 \alpha N$ copies of good $j$ are available, then
by Assumption~\ref{ass:min-bound}, there is an assignment with
valuation at least $\rho' \cdot \min\{N_j, \lambda^2 \alpha N\}$.
% if $\frac{}{l} \ge \lambda^2 \alpha N$ then, the social welfare is at least
% $\frac{2 \lambda + \lambda^2}{(1 + \lambda)^2} \cdot \lambda^2 \alpha N \cdot \rho'$.
% Otherwise, the social welfare is at least
% $\frac{2 \lambda + \lambda^2}{(1 + \lambda)^2} \cdot \frac{N}{l} \cdot \rho'$.
Therefore, the social welfare is at least
$\sum_j \min\{N_j, \lambda^2\alpha N\} \frac{2 \lambda + \lambda^2}{(1 + \lambda)^2}\cdot  \rho'
\ge \lambda^2\alpha  \frac{2 \lambda + \lambda^2}{(1 + \lambda)^2} N \cdot \rho' $.
\end{proof}

\subsection{Proofs from Section~\ref{sec:wal-equil}}

In Lemmas~\ref{epsilon-bad} and~\ref{k-epsilon-bad}, we bound the number of $(\epsilon,U)$-bad multiplicity
vectors, and then in Lemma~\ref{k-bad probability} we bound the probability of a $(k, \epsilon, U)$-bad vector. Following this, in Lemma~\ref{lowboundprice} and \ref{uppboundprice},
assuming the multiplicity vector is $(k+1, \epsilon, U)$-good,
we bound the difference between the English Walrasian mechanism prices and those of the Walrasian mechanism at hand. Next, in Lemma~\ref{smooth}, again for $(k+1,\epsilon, U)$-good multiplicity vectors, we relate $u_i(\bxi(v_i, b_{-i}))$ to $v_i(\bxi(v_i, v_{-i}))$ and the prices paid; we then use this to carry out a PoA analysis.

First, we have following two observations.
\begin{observation}\label{price upbound}
In the Dutch Walrasian mechanism, if there are zero copies of a good, letting its price be $+\infty$   will not affect the mechanism outcome.
\end{observation}

\begin{observation}\label{non-increasing}
Suppose bidders' demands satisfy the Gross Substitutes property. In both the English and Dutch Walrasian mechanisms, if $n_i \geq n'_i$,  then $\bp(n_i, n_{-i}) \preceq \bp(n'_i, n_{-i})$, where $\bp \preceq \bp'$ means that, for all $j$, $p_j \leq p'_j$.
\end{observation}
\begin{lemma}\label{epsilon-bad}
In the English Walrasian mechanism, given $n_{-j}$ and bidding profile $\bb$, the number of values $n_j$ for which $(n_j, n_{-j})$ is \textit{$(\epsilon, U)$-bad} for good $j$ is at most $\frac{m}{\epsilon}U$.
\hide{
\[
\left | \left \{n_j | dist(p^{Eng}((n_j, n_{-j}); \vec b), p^{Eng}((n_j + 1, n_{-j}); \vec b)) > \epsilon \right \} \right | \leq \frac{l}{\epsilon}.
\]
}
\end{lemma}

\pfof{Lemma~\ref{epsilon-bad}}
We prove the result by contradiction. Accordingly,
let\\
$S = \left \{n_j \Bigg| dist^U(\bp^{Eng}((n_j, n_{-j}); \bb), \bp^{Eng}((n_j + 1, n_{-j}); \bb)) > \epsilon \right \}$ and suppose that $|S| > \frac{m}{\epsilon}U$.
\[
\text{The proof uses a new function}~\mathit{pf}(\cdot):~~~~
\mathit{pf}(n_j) = \sum_{q = 1}^m \min\{p^{Eng}_q((n_j, n_{-j}); \bb),U\}.
\]
\begin{eqnarray} \label{label1}
\mbox{Then,} &&\liminf_{n \rightarrow \infty} (\mathit{pf}(0) - \mathit{pf}(n)) = \liminf_{n \rightarrow \infty} \sum_{h = 0}^{n-1} (\mathit{pf}(h) - \mathit{pf}(h+1)) \nonumber\\
&&~~~~~~\geq \sum_{n_j \in S}(\mathit{pf}(n_j) - \mathit{pf}(n_j + 1)) > \frac{m}{\epsilon} U \cdot \epsilon = m \cdot U.
\end{eqnarray}
The first inequality follows as by Observation \ref{non-increasing}, $\mathit{pf}(\cdot)$ is a non-increasing function.
Further, by construction, $0 \leq \mathit{pf}(h) \leq l \cdot U$ for all $h$, thus $\liminf_{n \rightarrow \infty} (\mathit{pf}(0) - \mathit{pf}(n)) \leq l \cdot U$, contradicting (\ref{label1}).
\end{proof}
\begin{lemma}\label{lem::math::1}
\[
{{m+n-1} \choose {n}} = \sum_{i = 0}^{n} {{m+i-2}\choose{i}}
\]
\end{lemma}

\begin{lemma}\label{lem::math::2}
\[\sum_{n = 0}^k {{m+n-1} \choose {n}} = \sum_{n = 0}^{k} (k - n + 1)  {{m+n-2}\choose{n}}.\]
\end{lemma}
\begin{proof}
\begin{align*}
\sum_{n = 0}^k {{m+n-1} \choose {n}} &=\sum_{n = 0}^{k} \sum_{i = 0}^{n} {{m+i-2}\choose{i}} = \sum_{i = 0}^{k} \sum_{n = i}^{k} {{m+i-2}\choose{i}} \\
&= \sum_{i = 0}^{k} (k - i + 1)  {{m+i-2}\choose{i}}.
\end{align*}
\end{proof}
\begin{lemma}\label{k-epsilon-bad}
In the English Walrasian mechanism with bidding profile $\bb$, for a fixed $n_{-j}$, the number of values $n_j$ for which $(n_j, n_{-j})$ is $(k, \epsilon, U)$-bad for good $j$ is at most $\frac{m}{\epsilon}U \cdot {{k+m} \choose m}$.
\end{lemma}
\pfof{Lemma~\ref{k-epsilon-bad}}
Consider the case that $m \geq 2$. For $(n_j, n_{-j})$ to be $(k, \epsilon, U)$-bad for good $j$
we need an $(\epsilon, U)$-bad vector $\bn' \preceq \bn$ for good $j$, with
$\sum_{h \neq j} n_h - n'_h = c$ for some $0 \le c \le k$ and $n_j - n'_j \le k-c$.
There are ${m-2+c} \choose  {c}$ ways of choosing the $n'_{-j}$.
For each $n'_{-j}$, by Lemma \ref{epsilon-bad}, 
there are at most $\frac{m}{\epsilon} U$ points that are $(\epsilon, U)$-bad for good $j$.
For each $(\epsilon, U)$-bad point, there are $k-c+1$ choices for $n_j$. 
This gives a total of
\[
\sum_{c=0}^k\frac {m}{\epsilon} U(k-c+1) {{m-2+c} \choose  {c}}  = \frac {m}{\epsilon} U \sum_{c = 0}^{k} {{m-1+c} \choose  {c}} = \frac{m}{\epsilon} U {{m+k} \choose  {k}}
\]
$(k, \epsilon, U)$-bad vectors. Note that the first equality follows by Lemma~\ref{lem::math::2} and the second equality follows by Lemma~\ref{lem::math::1}.

For the case $m = 1$, for each $(\epsilon, U)$-bad point for this good, it will cause at most $k + 1$ points to be $(k, \epsilon, U)$-bad for this good. This gives a total of 
\[
\frac{m}{\epsilon} U (k+1) = \frac{m}{\epsilon} U {{m+k} \choose  {k}}.
\]
$(k, \epsilon, U)$-bad vectors.
\hide{
If  $(n_j, n_{-j})$ is $(k, \epsilon)$-bad for good $j$, then there must be a point $(n'_j , n'_{-j})$ which is $\epsilon$-bad for good $j$,  with $n_q - k \leq n'_q \leq n_q$, for all $1 \leq q \leq l$. Fixing $n_{-j}$, there are $(k+1)^{l-1}$ possible values of $n'_{- j}$. For each $n'_{-j}$, by Lemma \ref{epsilon-bad}, there are at most $\frac{l}{\epsilon}$ points that are $\epsilon$-bad for good $j$.
Finally, for a fixed $n'_j$, there are $k+1$ possible choices of $n_j$. Thus, given $n_{-j}$, there are at most $(k+1)^{l-1} \frac{l}{\epsilon} (k+1) =  \frac{l}{\epsilon}(k+1)^l$ $(k, \epsilon)$-bad points for  good $j$.
}
\end{proof}

\pfof{Lemma~\ref{k-bad probability}}
Conditioned on the bidding profile being $\bb$,
\begin{align*}
% &\Pr[(\textit{ $\bn$ is  $(k, \epsilon)$-good for all goods}) \cap (min_{j} \{n_j\}  > k)] \\
&\sum_{1 \leq j \leq m} \Pr[(\textit{ $\bn$ is  $(k, \epsilon, U)$-bad for good $j$}) \cup (n_j \leq k)]) \\
&~~~~~~\leq \sum_{1 \leq j \leq m} \Pr[(\textit{$\bn$ is $(k, \epsilon, U)$-bad for good $j$})] + \Pr[(n_j \leq k)] \\
&~~~~~~\leq  \sum_{1 \leq j \leq m} \sum_{\overline{n}_{-j}}\Big(\Pr[(\textit{ $\bn$ is  $(k, \epsilon, U)$-bad for good $j$}) | n_{-j} = \overline{n}_{-j}] \\
&~~~~~~~~~~~~~~~~~~~~~~~~+ \Pr[(n_j \leq k) | n_{-j} = \overline{n}_{-j}]\Big) \cdot \Pr[n_{-j} = \overline{n}_{-j}] \\
&~~~~~~\leq  m F(N) \left[\frac{m}{\epsilon} U {{k+m} \choose {m}} + k+1\right] ~~~~~~ \mbox{(by Lemma \ref{k-epsilon-bad}).}
\end{align*}
%
% Since $\lim_{N \rightarrow \infty} f_j(N) = 0$ for all $i$, the probability converges to $1$ as $N \rightarrow \infty$.
\end{proof}

Let $ n^i_j(b_i, b_{-i}) $ denote the number of copies of good $j$ that bidder $i$ receives with bidding profile $(b_i, b_{-i})$ and $n^i(b_i, b_{-i})$ denote the corresponding vector.
% Let $\vec k = (k,k,\cdots,k)$ be an $l$-vector.
Also, let $p^{Eng}(\bn; b_{-i})$ denote the market equilibrium prices when bidder $i$ is not
present.

\begin{lemma}\label{lowboundprice}
$p^{Eng}_j(\bn; b_{-i}) \leq  p_j(\bn; (b_i, b_{-i})).$
\end{lemma}
\pfof{Lemma~\ref{lowboundprice}}
Consider the situation with $\bn' = \bn -  \bn^i (b_i, b_{-i})$ and suppose that agent $i$ is not present. Then $p_j(\bn; (b_i, b_{-i}))$ is a market equilibrium.
\begin{align*}
\text{So} \hspace*{1in}& \forall j ~~~ p^{Eng}_j(\bn'; b_{-i}) \leq  p_j(\bn; (b_i, b_{-i})). \\
\text{Since}~\bn \geq \bn',\hspace*{0.3in} &\forall j ~~~ p^{Eng}_j(\bn; b_{-i}) \leq p^{Eng}_j(\bn'; b_{-i}).\hspace*{2in}
\end{align*}

The lemma follows on combining these two inequalities.
\end{proof}

\begin{lemma}\label{uppboundprice}
If $\bn$ is $(k+1, \epsilon, U)$-good for all goods, and $n_j > k + 1$ for all $j$,  then
\[ \forall j~~~\min\{p_j(\bn; (v_i, b_{-i})), U\}  \leq \min\{p_j(\bn; (b_{i}, b_{-i})), U\} + (k + 1) \epsilon.\]
\end{lemma}
\pfof{Lemma~\ref{uppboundprice}}
Let $\bd^i \preceq \bn^i(v_i, b_{-i})$ be a minimal set with $v_i(\bd^i) = v_i(\bn^i(v_i,  b_{-i}))$. By Definition~\ref{large-market}(i), $\sum_j d^i_j \leq k$.
First, if $n^i_j(v_i, b_{-i}) > d^i_j$ then $p_j(\bn; (v_i, b_{-i})) = 0$, as the pricing is given by a Walrasian Mechanism.

Consider the scenario with $\bn'$ copies of goods on offer, where for all $j$, $n'_j =  n_j -  d^i_j $ and suppose that bidder $i$ is not present;  then $\bp(\bn; (v_i, b_{-i}))$ is a market equilibrium.
\begin{align*}
&\text{So,} &p_j(\bn; (v_i, b_{-i})) \leq p^{Dut}_j(\bn'; b_{-i}).\\
&\mbox{For all $j' \neq j$, let $n''_{j' } = n'_{j' }$ and let $n''_j = n'_j - 1$; then}~~~~~~ & p^{Dut}_j(\bn'; b_{-i}) \leq p^{Eng}_j(\bn''; b_{-i}),\\
&\text{and by Lemma~\ref{lowboundprice},}\hspace*{1.05in}&p^{Eng}_j(\bn''; b_{-i})\leq p^{Eng}_j(\bn''; {b_i, b_{-i}}).
\end{align*}
As $\bn$ is $(k+1, \epsilon, U)$-good for all goods, and as $\sum_h n_h - n''_h \leq k + 1$,  we conclude that
\begin{align}
&~~~\min\{p_j(\bn; (v_i, b_{-i})), U\} \leq \min\{p^{Eng}_j(\bn''; (b_i, b_{-i})), U\} \nonumber\\
 &~~~~~~\leq \min\{p^{Eng}_j(\bn; (b_i, b_{-i})), U\} + (k + 1) \epsilon \leq \min\{p_j(\bn; (b_i, b_{-i})), U\} + (k + 1) \epsilon.
\end{align}
\end{proof}

\pfof{Lemma~\ref{smooth}}
As we are using a Walrasian mechanism, for any allocation $x'_i$,
\begin{eqnarray}
\label{eqn:wal-value}
v_i(x_i(v_i, b_{-i})) - \sum_{s \in \bxi(v_i, b_{-i})} p_s(\bn; (v_i, b_{-i}))
\geq v_i(\bxi') - \sum_{s \in \bxi' } p_s(\bn; (v_i, b_{-i})).
\end{eqnarray}

We let $S$ denote the set of goods whose prices $p_s(\bn; (v_i, b_{-i}))$ are larger than $U$.  Then,
\begin{eqnarray} \label{smooth1}
u_i(v_i, b_{-i}) &=& v_i(\bxi(v_i, b_{-i})) - \sum_{s \in \bxi(v_i, b_{-i})} p_s(\bn; (v_i, b_{-i})) \nonumber\\
&\geq& v_i((\bxi(v_i, v_{-i}) \cap \bdi) \setminus S) - \sum_{s \in (\bxi(v_i, v_{-i}) \cap \bdi) \setminus S} p_s(\bn; (v_i, b_{-i}))~~~~\text{(by~\eqref{eqn:wal-value})}
\end{eqnarray}

Since $\bn$ is $(k+1, \epsilon, U)$-good, by Lemma~\ref{uppboundprice},
\[\min\{p_s(\bn; (v_i, b_{-i})),U\} \leq \min\{p_s(\bn; (b_i, b_{-i})),U\} + (k + 1) \epsilon.\]
\begin{align*}
\text{Therefore, for any}~ s \notin S,~~~~~
p_s(\bn; (v_i, b_{-i})) &\leq \min\{p_s(\bn; (b_i, b_{-i})),U\} + (k + 1) \epsilon \\
&\leq p_s(\bn; (b_i, b_{-i})) + (k + 1)\epsilon.
\end{align*}
\begin{eqnarray}\label{smooth2}
\text{So,}~~~~&&v_i((\bxi(v_i, v_{-i}) \cap \bdi) \setminus S) - \sum_{s \in (\bxi(v_i, v_{-i})  \cap \bdi) \setminus S} p_s(\bn; (v_i, b_{-i})) \nonumber\\
&&~~~~~~\geq v_i((\bxi(v_i, v_{-i}) \cap \bdi) \setminus S) - \sum_{s \in (\bxi(v_i, v_{-i}) \cap \bdi) \setminus S} p_s(\bn; (b_i, b_{-i})) \nonumber\\
&&~~~~~~~~~~~~ - \left | (\bxi(v_i, v_{-i}) \cap \bdi) \setminus S \right | \cdot (k+1) \epsilon.
\end{eqnarray}

For any $s \in S$, on applying Lemma~\ref{uppboundprice},
we obtain $U = \min\{p_s(\bn; (v_i, b_{-i})),U\} \leq \min\{p_s(\bn; (b_i, b_{-i})),U\} + (k + 1) \epsilon$, which implies $p_s(\bn; (b_i, b_{-i})) + (k + 1) \epsilon \geq U$. Also,
\[
v_i(\bxi(v_i, v_{-i}) \cap \bdi) -  v_i((\bxi(v_i, v_{-i}) \cap \bdi) \setminus S) \leq
v_i((\bxi(v_i, v_{-i}) \cap \bdi) \cap S) \leq |(\bxi(v_i, v_{-i}) \cap \bdi) \cap S| \cdot U,
\]
where the first inequality follows by the Gross Substitutes assumption,
and the second by Gross Substitutes and because by assumption $v_i(s) \leq U$ for all single items $s$.
Thus,
\begin{align}\label{smooth3}
&v_i((\bxi(v_i, v_{-i}) \cap \bdi) \setminus S) - \sum_{s \in (\bxi(v_i, v_{-i}) \cap \bdi) \setminus S} p_s(\bn; (b_i, b_{-i})) - \left | (\bxi(v_i, v_{-i}) \cap \bdi) \setminus S \right | \cdot (k+1) \epsilon \nonumber\\
&~~~\geq v_i(\bxi(v_i, v_{-i}) \cap \bdi) - |(\bxi(v_i, v_{-i}) \cap \bdi) \cap S| \cdot U \nonumber \\
& \hspace*{1.5in}- \sum_{s \in (\bxi(v_i, v_{-i}) \cap \bdi)\setminus S} p_s(\bn; (b_i, b_{-i})) - \left | (\bxi(v_i, v_{-i}) \cap \bdi) \setminus S  \right | \cdot (k+1) \epsilon  \nonumber\\
&~~~\geq v_i(\bxi(v_i, v_{-i}) \cap \bdi) - \sum_{s \in \bxi(v_i, v_{-i}) \cap \bdi} p_s(\bn; (b_i, b_{-i})) - \left | \bxi(v_i, v_{-i}) \cap \bdi  \right | \cdot (k+1) \epsilon  \nonumber\\
&~~~\geq v_i(\bxi(v_i, v_{-i})) - \sum_{s \in \bxi(v_i, v_{-i})} p_s(\bn; (b_i, b_{-i})) - \left | \bxi(v_i, v_{-i}) \cap \bdi  \right | \cdot (k+1) \epsilon.  \nonumber\\
\end{align}

By (\ref{smooth1}), (\ref{smooth2}) and (\ref{smooth3}),
\begin{eqnarray}
u_i(v_i, b_{-i}) &\geq& v_i(\bxi(v_i, v_{-i})) - \sum_{s \in \bxi(v_i, v_{-i})} p_s(\bn; (b_i, b_{-i})) - \left | \bxi(v_i, v_{-i}) \cap \bdi  \right | \cdot (k+1) \epsilon. \nonumber
\end{eqnarray}
\end{proof}

\pfof{Theorem~\ref{wal-e1q-PoA}}

By Lemma~\ref{k-bad probability}, the probability that $\bn$ is $(k+1,  \frac{\max_s\{v_i(s)\}}{2^c}, \max_s\{v_i(s)\})$-bad or $n_j \leq k+1$ for some $j$ is less than 
\begin{eqnarray}
&&\frac{m}{L} \left[ \frac{1}{\frac{\max_s\{v_i(s)\}}{2^c}}\max_s\{v_i(s)\}\Lambda(m, k+1) + k+2\right] \nonumber\\
&&~~~~~~ = \frac{m}{L} \left[ 2^c \Lambda(m, k+1) + k+2\right]. \nonumber
\end{eqnarray} 
So, for any integer $c' $,
\begin{align*}
\mathbb{E}_{\bn}[u_i(v_i, b_{-i}) ] &\geq \mathbb{E}_{\bn}\Bigg[v_i(\bxi(v_i, v_{-i})) - \sum_{s \in \bxi(v_i, v_{-i})} p_s(\bn; (b_i, b_{-i})) \\
&~~~~~~~~~ - \sum_{c = 1}^{c'} \mathds{1}\Big[\text{$\bn$ is $(k+1,\frac{\max_s\{v_i(s)\}}{2^c},\max_s\{v_i(s)\})$-bad and} \nonumber\\
&~~~~~~~~~~~~~~~~~~~~~~~~~~~~~~~~~~~ \text{$(k+1,\frac{\max_s\{v_i(s)\}}{2^{c-1}}, \max_s\{v_i(s)\})$-good}\Big] \\
&~~~~~~~~~~~~~~~~~~\cdot  \Big| \bxi(v_i, v_{-i})  \cap \bdi \Big| \cdot (k+1) \frac{\max_s\{v_i(s)\}}{2^{c-1}}   \\
&~~~~~~~~~ - \mathds{1}\Big[\text{$\bn$ is $(k+1,\frac{\max_s\{v_i(s)\}}{2^{c'-1}}, \max_s\{v_i(s)\})$-good}\Big] \nonumber\\
&~~~~~~~~~~~~~~~~~~~~~~~~~~~~~~~~~~~ \cdot  \Big| \bxi(v_i, v_{-i})  \cap \bdi \Big| \cdot (k+1) \frac{\max_s\{v_i(s)\}}{2^{c'}}  \Bigg] \\
&\geq \mathbb{E}_{\bn}\Bigg[v_i(\bxi(v_i, v_{-i})) - \sum_{s \in \bxi(v_i, v_{-i})} p_s(\bn; (b_i, b_{-i})) \Bigg] \\
&~~~~~~~~~ - \sum_{c = 1}^{c'} \frac{m}{L} \left[2^c \Lambda(m, k+1) + k+2\right] \cdot  k \cdot (k+1) \frac{\max_s\{v_i(s)\}}{2^{c-1}} \\
&~~~~~~~~~ - k \cdot (k+1) \frac{\max_s\{v_i(s)\}}{2^{c'}}  \\
&\geq \mathbb{E}_{\bn}\Bigg[v_i(\bxi(v_i, v_{-i})) - \sum_{s \in \bxi(v_i, v_{-i})} p_s(\bn; (b_i, b_{-i})) \\
&~~~~~~~~~ - c' \cdot \frac{m}{L} \left[ 2 \Lambda(m, k+1) + k+2\right] \cdot  k \cdot (k+1) \cdot \max_s\{v_i(s)\}    \\
&~~~~~~~~~ - k \cdot (k+1) \frac{\max_s\{v_i(s)\}}{2^{c'}} \Bigg]. \\
\end{align*}
Summing over all the bidders and integrating w.r.t. $\bv$ and $\bb$ gives
\begin{align*}
\sum_i \mathbb{E}_{\bv, \bb, \bn}[u_i(v_i, b_{-i})] &\geq  \sw(\opt) - \mathbb{E}_{\bb}[\text{R}(b_i, b_{-i})] \\
&~~~~~~ - N \cdot c' \cdot \frac{m}{L} \left[ 2 \Lambda(l, k+1) + k+2\right] \cdot k \cdot (k + 1) \cdot \zeta \cdot m \\
&~~~~~~ - N \cdot k \cdot (k + 1) \frac{1}{2^{c'}} \cdot \zeta \cdot m.
\end{align*}

Using the smooth technique for Bayesian settings~\cite{Syrgkanis:2013:CEM:2488608.2488635} yields
\begin{eqnarray}
\text{SW(NE)} &\geq&  \Bigg(1 - \frac{\zeta \cdot m \cdot k \cdot (k + 1) \frac{1}{2^{c'}}}{\rho} \nonumber \\
&&~~~~~~ - \frac{\zeta \cdot m \cdot  c' \cdot \frac{m}{L} \left[ 2 \Lambda(m, k+1) + k+2\right] \cdot k \cdot (k + 1)}{\rho}\Bigg)  \sw(\opt). \nonumber
\end{eqnarray}

Let $Y = \frac{m}{L} \left[ 2 \Lambda(m, k+1) \right]$. Set $c' = \lceil \log_2 \frac{1}{Y} - \log_2 \log_2 \frac{1}{Y} \rceil$; then $\frac{1}{2^{c'}} \leq Y \log_2  \frac{1}{Y}$. So, 
\begin{eqnarray}
\text{SW(NE)} \geq \Bigg(1 - \frac{3 \cdot k \cdot (k+1) \cdot \zeta \cdot m}{\rho} \cdot  Y \cdot  \lceil \log_2 \frac{1}{Y}\rceil  \Bigg) \sw(\opt). \nonumber
\end{eqnarray}
\end{proof}

\section{Regret Minimization}
\subsection{Walrasian Market}

We note the following corollary to Theorem~\ref{wal-e1q-PoA}.
\begin{corollary}
\label{col:regret-wal}
In a large Walrasian auction which satisfies Assumptions~\ref{ass:bdd-val}, if $v_i$ and $b_i$ are monotone and  satisfy the gross substitutes property for all $i$, then
% the Bounded Valuation and Market Size properties,
% If $\text{SW(OPT)} > \rho N$, then
\[
\sum_i \mathbb{E}_{\bn, \bv, \bb} [u_i(v_i, b_{-i})] \geq \Bigg(1 - \frac{3 \cdot k \cdot (k+1) \cdot \zeta \cdot m}{\rho} \cdot  Y \cdot  \lceil \log_2 \frac{1}{Y}\rceil \Bigg) \sw(\opt) - \mathbb{E}_{\bb}[\text{R}(b_i, b_{-i})]   
\]
where $Y = \frac{m}{L} \left[ 2 m {{k+1+m}\choose{m}} \right]$ and $\rho = \frac{\sw(\opt)}{N}$.
\end{corollary}

\pfof{Theorem~\ref{thm:regret-wal}}
Since player $i$ uses a regret minimizing algorithm and she is a $(\gamma, \delta)$-player,
\[
\mathbb{E}_{\bn} \Bigg[ \sum_{t = 1}^{T} v_i (b_{i}^t, b_{-i}^t) \Bigg] \geq  \mathbb{E}_{\bn} \Bigg[ \sum_{t = 1}^{T} u_i (v_{i}, b_{-i}^t)  - \Phi(|K_i|,T) \cdot (\max_{\bxi} v_i(\bxi) \cdot \gamma + \delta) \Bigg].
\]

Summing over all bidders and integrating w.r.t. $\bv$ and $\bb$ gives
\begin{align*}
\mathbb{E}_{\bn, \bv, \bb} \Bigg[\sum_{i} \sum_{t = 1}^{T} u_i (b_{i}^t, b_{-i}^t)\Bigg] &\geq  \mathbb{E}_{\bn, \bv, \bb} \Bigg[\sum_{i} \sum_{t = 1}^{T} u_i (v_{i}, b_{-i}^t)  - \Phi(|K_i|,T) \cdot (\max_{\bxi} v_i(\bxi) \cdot  \gamma + \delta)\Bigg] \\
&\geq  \mathbb{E}_{\bn, \bv, \bb} \Bigg[\sum_{i} \sum_{t = 1}^{T} u_i (v_{i}, b_{-i}^t)  - \Phi(|K_i|,T) \cdot (k m \zeta  \gamma + \delta)\Bigg].
\end{align*}

By Corollary~\ref{col:regret-wal}, 
\begin{eqnarray}
\label{reg-wal}
\sum_i \mathbb{E}_{\bn, \bv, \bb} [u_i(v_i, b_{-i})] &\geq& \Bigg(1 - \frac{3 \cdot k \cdot (k+1) \cdot \zeta \cdot m}{\rho} \cdot  Y \cdot  \lceil \log_2 \frac{1}{Y}\rceil \Bigg) \sw(\opt) \nonumber\\
&&~~~~~~~~~- \mathbb{E}_{\bb}[\text{R}(b_i, b_{-i})].
\end{eqnarray}

Therefore, since valuation equals utility plus payment,
\begin{align*}
&\mathbb{E}_{\bn, \bv, \bb} \Bigg[\frac{1}{T}\sum_{i} \sum_{t = 1}^{T} v_i (\bxi(b_{i}^t, b_{-i}^t))\Bigg]\\
 &~~~~~~=  \mathbb{E}_{\bn, \bv, \bb} \Bigg[\frac{1}{T}\sum_{i} \sum_{t = 1}^{T} (u_i (b_{i}^t, b_{-i}^t) + \text{R}(b_i^t, b_{-i}^t))\Bigg] \\
&~~~~~~\geq  \frac{1}{T}\mathbb{E}_{\bn, \bv, \bb} \Bigg[\sum_{i} \Bigg(\sum_{t = 1}^{T} (u_i (v_{i}, b_{-i}^t)  +  \text{R}(b_i^t, b_{-i}^t)) - \Phi(|K_i|,T) \cdot (k m \zeta  \gamma + \delta)\Bigg) \Bigg] \\
&~~~~~~\geq \Bigg(1 - \frac{3 \cdot k \cdot (k+1) \cdot \zeta \cdot m}{\rho} \cdot  Y \cdot  \lceil \log_2 \frac{1}{Y}\rceil \Bigg) \sw(\opt) \\
&~~~~~~~~~~~~~~~~~~- \frac{1}{T}\sum_{i} \Phi(|K_i|,T) \cdot (k m \zeta  \gamma + \delta) ~~~~~~~~\mbox{by \eqref{reg-wal}}\\
&~~~~~~\geq \Bigg(1 - \frac{3 \cdot k \cdot (k+1) \cdot \zeta \cdot m}{\rho} \cdot  Y \cdot  \lceil \log_2 \frac{1}{Y}\rceil \Bigg) \sw(\opt) \\
&~~~~~~~~~~~~~~~~~~- \max_{i} \Phi(|K_i|,T) \cdot (k m \zeta  \gamma + \delta) \frac{1}{\rho \cdot T} \sw(\opt) \\
&~~~~~~= \Bigg(1 - \frac{3 \cdot k \cdot (k+1) \cdot \zeta \cdot m}{\rho} \cdot  Y \cdot  \lceil \log_2 \frac{1}{Y}\rceil \\
&~~~~~~~~~~~~~~~~~~- \frac{\max_{i} \Phi(|K_i|,T) \cdot (k m \zeta  \gamma + \delta)}{\rho \cdot T} \Bigg) \sw(\opt).
\end{align*}
\end{proof}

\subsection{Fisher Market with Reserve Prices}

Theorem~\ref{thm:regret-fisher} will follow from the following lemma;  its proof  is given in Appendix~\ref{sec:reserve-prices}.

\begin{theorem}
\label{thm:FisherPoAReserve}
For any bidding profile $\mathbf{b}$ and any value profile $\mathbf{v}$ which are homogeneous of degree 1, concave, continuous, monotone and gross substitutes, if the reserve prices $r_j \leq \frac{1}{4} p^*_j$ for any $j$, then
\[
\sum_i \ui(\vi, \bmi)) \geq  e^{-\frac{2m}{5L}} \sum_i u_i(\bxi(\bp^*)).
\]
\end{theorem}

\pfof{Theorem~\ref{thm:regret-fisher}}
Since player $i$ uses a regret minimizing algorithm  and the maximal payoff is $\lambda u_i(v_i, v_{-i})$,
\begin{align*}
\sum_{t = 1}^{T} u_i (b_{i}^t, b_{-i}^t) &\geq \sum_{t = 1}^{T} u_i (\vi, b_{-i}^t) - \Phi(|K_i|,T) \cdot \lambda u_i(v_i, v_{-i}).
\end{align*}

Summing over all the bidders gives
\begin{align*}
\sum_i \sum_{t = 1}^{T}  u_i (b_{i}^t, b_{-i}^t) &\geq \sum_i \sum_{t = 1}^{T} u_i (\vi, b_{-i}^t) - \sum_i \Phi(|K_i|,T) \cdot \lambda u_i(v_i, v_{-i}) \\
&\geq \sum_i \sum_{t = 1}^T u_i (\vi, b_{-i}^t) - \sum_i \max_{i'} \Phi(|K_{i'}|, T) \cdot \lambda u_i(v_i, v_{-i}).
\end{align*}

By Theorem~\ref{thm:FisherPoAReserve}, 
\[
\sum_i \ui(\vi, \bmi)) \geq  e^{-\frac{2m}{L}} \sum_i u_i(\bxi(\bp^*)).
\]

Therefore, 
\begin{align*}
\sum_i \sum_{t = 1}^{T}  u_i (b_{i}^t, b_{-i}^t) &\geq \sum_i \sum_{t = 1}^T u_i (\vi, b_{-i}^t) - \sum_i \max_{i'} \Phi(|K_{i'}|, T) \cdot \lambda u_i(v_i, v_{-i}) \\
& \geq T \cdot  e^{-\frac{2m}{L}} \sum_i u_i(\bxi(\bp^*)) - \max_{i'} \Phi(|K_{i'}|, T)  \lambda \sum_i u_i(v_i, v_{-i}).
\end{align*}

The theorem follows on dividing both sides by $T$.
\end{proof}

\section{Reserve prices}
\label{sec:reserve-prices}
\begin{definition}
Eisenberg Gale markets with reserve prices are exactly the solutions to the following convex program:
\begin{align*}
\max_{\bx}  ~~\sum_{i=1}^{n} & e_i \cdot \log(u_i(x_{i1}, x_{i2}, \cdots x_{im})) + \sum_{j = 1}^{m} y_j r_j\nonumber\\
&\text{s.t.} ~~~~ \forall j:~~~~ \sum_i x_{ij}  + y_j\leq 1 \nonumber\\
& ~~~~~~~~~\forall i, j:~~ x_{ij} \geq 0, \nonumber
\end{align*}
where $r_j$ is the reserve price of item $j$.
\end{definition}

The proof of Theorem~\ref{thm:FisherPoAReserve} uses the following lemma.
\begin{lemma}
\label{lem:final-bound}
For any bidding profile $\mathbf{b}$ and any value profile $\mathbf{v}$ which are homogeneous of degree 1, concave, continuous, monotone and satisfy the gross substitutes property, if the reserve prices $r_j \leq \frac{1}{4} p^*_j$ for any $j$, then
\[
\sum_i \ei \log \ui(\vi, \bmi)) - \sum_i \ei \log u_i(\bxi(\bp^*))  \ge -2 m \cdot \max_{i'} e_{i'}.
\]
\end{lemma}

\pfof{Theorem~\ref{thm:FisherPoAReserve}}
On exponentiating the expressions on both sides in the statement of Lemma~\ref{lem:final-bound}
we obtain
\begin{eqnarray}
\prod_i u_i(v_i, b_{-i})^{e_i} \geq \frac{1}{e^{2 m \cdot \max_i e_i}} \prod_i u_i(v_i, v_{-i})^{e_i}.  \nonumber
\end{eqnarray}

Therefore,
\begin{eqnarray}
\prod_i \Big( \frac{u_i(v_i, b_{-i})}{u_i(v_i, v_{-i})} \Big)^{e_i} \geq  \frac{1}{e^{2 m \cdot \max_i e_i}}. \nonumber
\end{eqnarray}

Using the weighted GM-AM inequality, we obtain
\begin{eqnarray}
\frac{\sum_i e_i  \frac{u_i(v_i, b_{-i})}{u_i(v_i, v_{-i})} }{ \sum_i e_i  } \geq \Bigg(\prod_i \Big( \frac{u_i(v_i, b_{-i})}{u_i(v_i, v_{-i})} \Big)^{e_i}\Bigg)^{\frac{1}{\sum_i e_i}} \geq \Bigg(\frac{1}{e^{2 m \max_i e_i }}\Bigg)^{\frac{1}{\sum_i e_i}} = e^{-\frac{2m \max_i e_i}{\sum_i e_i}}. \nonumber
\end{eqnarray}

Since $u_i(v_i, v_{-i}) = t e_i$, for all $i$, 
\begin{align*}
\sum_i u_i(v_i, b_{-i}) \geq e^{-\frac{2m \max_i e_i}{\sum_i e_i}} \sum_i u_i(v_i, v_{-i}). \nonumber
\end{align*}
\end{proof}

Our goal is to bound $\sum_i \ei \log \ui(\vi, \vmi) - \sum_i \ei \log \ui(\vi, \bmi)$.
We will be working with the following function, the demand at prices $\bp$:
\begin{equation}
\label{eqn:x-of-p-defn}
 x(\mathbf{p}) = (x_1(\mathbf{p}), x_2(\mathbf{p}), \cdots) 
= \arg \max_{\bx} \sum_i e_i \log u_i(\bxi) + \mathsf{1} \cdot \mathbf{p} - \sum_i \bxi \cdot \mathbf{p}.
\end{equation}
Recall that, by Lemma~\ref{lem::price}, $\pj(\vi,\bmi) \le \pj(\bi,\bmi) +\mathsf{1} \cdot e_i$.
Consequently, $\ui(\exi(\bp(\vi, \bmi))) \ge \ui(\exi(\bp(\bb)+\mathsf{1} \cdot \max_{i'} e_{i'}))$, and so it will suffice to
bound $\sum_i \ei \ui(\vi, \vmi)) - \sum_i \ei \ui(\exi(\bp(\bb)+\mathsf{1} \cdot \max_{i'} e_{i'}))$.

We want to apply the bound in~\eqref{multi-poa-ineq}, but then we need prices $\bq$
such that $\sum_j q_j = \sum_i \ei$.
Accordingly, we will be considering the \emph{scaled} prices $\bq(\bb) = (\bp(\bb)+\mathsf{1} \cdot \max_{i'} e_{i'})\cdot \frac{ \sum_i \exi}{\sum_j \pj(\bb) + \max_{i'} e_{i'}}$
and the \emph{compressed} prices, defined below.

For convenience, in the following definition, we set $\frac{0}{x} = 0$ for $x > 0$, $\frac{0}{0} = 1$ and $\frac{x}{0} = +\infty$ for $x > 0$.
\begin{definition}
\label{eqn:compressed-prices}
Let $\bq$ be a price vector such that $\mathsf{1} \cdot \bq = \sum_i e_i$. The $l$-compressed version  ($l \leq 1$) of $\bq$ is
defined as
$p'_j (l, \bq)$ where 
\begin{align*}
&\frac{p'_j(l, \bq)}{p^*_j} = l ~~~~~~~~~\mbox{if}~~~ \frac{q_j}{p^*_j} \leq l,  \\
&\frac{p'_j(l, \bq)}{p^*_j} = t ~~~~~~~~~\mbox{if} ~~~\frac{q_j}{p^*_j} \geq t,   \\
\mbox{and}~~~~&\frac{p'_j(l, \bq)}{p^*_j} = \frac{q_j}{p^*_j}~~~~~\mbox{ if}~~~ l < \frac{q_j}{p^*_j} < t, 
\end{align*}
where $t$ is a number bigger than $1$ 
such that $\sum_j p'_j (l, \bq) = \sum_i e_i$, and $\mathbf{p}^*$ is the optimal solution $(\mathsf{1} \cdot \bp^* = \sum_i e_i)$. 
\end{definition}

Henceforth, unless noted otherwise, we let $\bp$ denote $\bp(\bb)$ and $\bq$ denote $\bq(\bb)$.

\begin{lemma}
\label{lem:log-relationship}
\begin{align*}
\sum_i \ei \log \ui(\exi(\bp + \mathsf{1} \cdot \max_{i'} e_{i'}))
& = \sum_i \ei \log \ui \left(\exi(\bp+ \mathsf{1} \cdot \max_{i'} e_{i'}) \cdot \frac{\sum_i \ei} {\sum_j (\pj + \max_{i'} e_{i'})} \right) \\
&~~~~~~~~~~~~~~~   - \sum_i \ei \log \frac {\sum_j (\pj + \max_{i'} e_{i'})} {\sum_i \ei} \\
& = \sum_i \ei \log \ui (\bq) - \sum_i \ei \log \frac {\sum_j (\pj + \max_{i'} e_{i'})} {\sum_i \ei}.
\end{align*}
\end{lemma}
\begin{proof}
\begin{align*}
\sum_i e_i \log u_i(x_i(\mathbf{p} + \mathsf{1}\cdot \max_{i'} e_{i'}))  &= \sum_i e_i \log \Big[ u_i(\bxi(\mathbf{p} + \mathsf{1} \cdot \max_{i'} e_{i'}) \\
&~~~~~~~~~~~~~~~\cdot \frac{\sum_j (p_j + \max_{i'} e_{i'})}{\sum_i e_i } \cdot \frac{\sum_i e_i }{\sum_j (p_j + \max_{i'} e_{i'})})\Big] \\
&= \sum_i e_i \log \Big[u_i(\bxi((\mathbf{p} + \mathsf{1}\cdot \max_{i'} e_{i'})\\
&~~~~~~~~~~~~~~~\cdot \frac{\sum_i e_i }{\sum_j p_j + \max_{i'} e_{i'}}) ) \cdot \frac{\sum_j e_i }{\sum_j p_j + \max_{i'} e_{i'}}\Big].  
\end{align*}

Now
\begin{align*}
&\sum_i e_i \log \Big[ u_i(\bxi((\mathbf{p} + \mathsf{1}\cdot \max_{i'} e_{i'})\cdot \frac{\sum_i e_i }{\sum_j p_j + \max_{i'} e_{i'}}) ) \cdot \frac{\sum_j e_i }{\sum_j p_j + \max_{i'} e_{i'}}\Big] \\
&~~~~~~= \sum_i e_i \log u_i(\bxi((\mathbf{p} + \mathsf{1}\cdot \max_{i'} e_{i'})\cdot \frac{\sum_i e_i }{\sum_j p_j + \max_{i'} e_{i'}}) )  - \sum_i e_i \log  \frac{\sum_j (p_j + \max_{i'} e_{i'})}{\sum_i e_i }. 
\end{align*}
\end{proof}

\begin{lemma}
\label{lem:p-pprime-comp}
Suppose that $\sum_j \qj = \sum_j \pj' = \sum_i \ei$.
Then
\[
 \sum_i e_i \log u_i(\bxi(\mathbf{q})) - \sum_i e_i \log u_i(\bxi(\mathbf{p}'))
\ge \sum_{ij} (p'_j - q_j) x_{ij}(\mathbf{p'}).
\]
\end{lemma}
\begin{proof}
As $\mathbf{x}(\mathbf{q}) =\arg \max_{\bx} \sum_i e_i \log u_i(\bxi) - \sum_i \bxi \cdot \mathbf{q} + \mathsf{1} \cdot \mathbf{q}$,
\begin{align*}
& \sum_i e_i \log u_i(\bxi(\mathbf{q})) - \sum_i e_i \log u_i(\bxi(\mathbf{p}')) \\
& ~~~~~~~\geq \sum_i \mathbf{q} \cdot \bxi(\mathbf{q}) - \sum_i \mathbf{q} \cdot \bxi(\mathbf{p}').
\end{align*}

As in the ``PoA'' analysis, $\bxi(\mathbf{q}) = \arg \max_{\bxi \cdot \mathbf{q} = e_i} u_i(\bxi)$. 
So,$\bxi(\mathbf{q}) \cdot \mathbf{q} = e_i$ and $\bxi(\mathbf{p}') \cdot \mathbf{p}' = e_i$. Therefore,
\begin{align*}
&\sum_i \mathbf{q} \cdot \bxi(\mathbf{q}) - \sum_i \mathbf{q} \cdot \bxi(\mathbf{p}') \\
& = \sum_i \mathbf{q} \cdot \bxi(\mathbf{q}) - \sum_i \mathbf{q} \cdot \bxi(\mathbf{p}') 
+ \sum_i \mathbf{p}' \cdot \bxi(\mathbf{p}') - \sum_i \mathbf{p}' \cdot \bxi(\mathbf{p}') \\
& = \sum_{ij} (p'_j - q_j) x_{ij}(\mathbf{p'}).
\end{align*}
\end{proof}

\begin{lemma} \label{lem::reserve::1}
There exists an $\bx(\mathbf{p}')$, where $\mathbf{p'} = p'_j(l, \mathbf{q})$, such that, for any $l < 1$, if $\frac{p'_j(l, \mathbf{q})}{p^*_j} = l  $, $\sum_i x_{ij}(\mathbf{p}') \geq \frac{1}{l}$ 
and if $\frac{p'_j(l, \mathbf{q})}{p^*_j} = t  $, $\sum_i x_{ij}(\mathbf{p}') \leq \frac{1}{t}$.
\end{lemma}
\begin{proof}
It is straightforward to check this for linear utility functions.

Now we consider the single-demand WGS utility functions.
Let $\widehat{\bp}$ denote the prices such that $\widehat{p}_j = p_j$ when $p_j > 0$ and $\widehat{p}_j = \epsilon$ when $p_j = 0$. Note that here $\epsilon$ is an arbitrarily small positive value.

By the homogeneity of the utility function,  there exists  an $\bx(l \mathbf{p}^*)$, such that $\sum_{i}x_{ij}(l \mathbf{p}^*) = \frac{1}{l}$ for all $j$ such that $p^*_j > 0$. Now, we consider $\sum_{i} x_{ij}(\widehat{l \mathbf{p}^*})$. For those $j$ such that $p^*_j > 0$, by the Gross Substitutes property, $\sum_{i} x_{ij}(\widehat{l \mathbf{p}^*})\geq \frac{1}{l}$.

Then, we let the price increase from $\widehat{l \mathbf{p}^*}$ to $\widehat{p'_j(l, \mathbf{q})}$. Also, by Gross Substitutes property, $\sum_{i} x_{ij}(\widehat{p'_j(l, \mathbf{q})})\geq \frac{1}{l}$ for those $j$ such that $p^*_j > 0$ and $\frac{p'_j(l, \mathbf{q})}{ p^*_j} = l $. By the same reasoning, $\sum_{i} x_{ij}(\widehat{p'_j(l, \mathbf{q})}) \leq \frac{1}{t}$ for those $j$ such that $p^*_j > 0$ and $\frac{p'_j(l, \mathbf{q})}{p^*_j} = t$.

Furthermore, by the Gross Substitutes property and homogeneity of the utility function, $\sum_{i} x_{ij}(\widehat{p'_j(l, \mathbf{q})}) = 0$ for those $j$ such that $p^*_j = 0$. \footnote{We consider a procedure that changes prices from $p^*$ to $\widehat{p'_j(l, \mathbf{q})}$. First, we define prices $\bp^{*'}$ such that $\frac{\widehat{p'_j(l, \mathbf{q})}}{p^{*'}_j} = k$ and $p^{*'}_j > 0$ if $p^*_j = 0$, and $\frac{\widehat{p'_j(l, \mathbf{q})}}{p^{*'}_j} \leq k$ and $p^{*'}_j = p^*_j$ if $p^*_j > 0$. Here, $k$ is a positive constant and $k$ is not infinity.  We increase the prices from $\bp^*$ to $\bp^{*'}$, by the Gross Substitutes property, $\sum_i x_{ij} (\bp^{*'}) = 0$ for those $j$ such that $p^*_j = 0$. Then, by homogeneity of the utility function, also for those $j$, $\sum_i x_{ij}(k \bp^{*'}) = 0$ . Now, we reduce the prices from $k \bp^{*'}$ to $\widehat{p'_j(l, \mathbf{q})}$ ($k \bp^{*'}$ is no less than $\widehat{p'_j(l, \mathbf{q})}$ by the definition of $\bp^{*'}$). Since $\widehat{p'_j(l, \mathbf{q})} = k p^{*'}_j$ for those $j$ such that $p^*_j = 0$, by Gross Substitutes property,  for those $j$, $\sum_i x_{ij} (\widehat{\bp'(l, \mathbf{q})}) \leq \sum_i x_{ij}(k \bp^{*'})$. By previous argument that $\sum_i x_{ij}(k \bp^{*'}) = 0$ for those $j$, $\sum_i x_{ij} (\widehat{\bp'(l, \mathbf{q})}) = 0$ also holds for those $j$.}

So, there exists an $\bx(\widehat{\mathbf{p}'})$, where $\mathbf{p'} = p'_j(l, \mathbf{q})$, such that for $p^*_j > 0$, if $\frac{p'_j(l, \mathbf{q})}{p^*_j} = l  $, $\sum_i x_{ij}(\widehat{\mathbf{p}'}) \geq \frac{1}{l}$ 
and if $\frac{p'_j(l, \mathbf{q})}{p^*_j} = t $, $\sum_i x_{ij}(\widehat{\mathbf{p}'}) \leq \frac{1}{t}$, and for $p^*_j = 0$, $\sum_i x_{ij}(\widehat{\mathbf{p}'}) = 0$.

Since $l < 1$, $\frac{p'_j(l, \mathbf{q})}{p^*_j} \neq l$ when $p^*_j = 0$. Therefore, we have an $\bx(\widehat{\mathbf{p}'})$, where $\mathbf{p'} = p'_j(l, \mathbf{q})$, such that if $\frac{p'_j(l, \mathbf{q})}{p^*_j} = l  $, $\sum_i x_{ij}(\widehat{\mathbf{p}'}) \geq \frac{1}{l}$ 
and if $\frac{p'_j(l, \mathbf{q})}{p^*_j} = t $, $\sum_i x_{ij}(\widehat{\mathbf{p}'}) \leq \frac{1}{t}$.

By the Gross Substitutes property, the demand $\bxi(\widehat{\bp'})$ for a given $\epsilon$ is also an optimal demand for any $0 < \epsilon' < \epsilon$. This is because for any small positive $\epsilon$, the demand for those goods with price $\epsilon$ is $0$. For reducing prices on those price $\epsilon$ goods only reduces the demand for other goods, but as there can be no  reduction in spending on the latter goods, in fact the demands are unchanged.

Therefore, by the continuity and homogeneity of the utility function, there exists an optimal allocation $\bx(\mathbf{p'})$,  which equals $\bx(\widehat{\bp'})$, where $\mathbf{p'} = p'_j(l, \mathbf{q})$. For if not, suppose there were a higher utility allocation when $\epsilon = 0$, whose value is $u_0$, where $u_0 >  u_i(\bxi(\widehat{\bp'}))$. Then at any $\epsilon > 0$, we could achieve utility $(1 - \lambda(\epsilon)) u_0$, where $\lim_{\epsilon \rightarrow 0} \lambda(\epsilon) = 0$, and as $x_{ij}(\widehat{\bp'})$ is an optimal allocation, $u_i(\bxi(\widehat{\bp'})) \geq (1 - \lambda(\epsilon)) u_0$. But this holds for every $\epsilon > 0$, so letting  $\epsilon \rightarrow 0$ yields $u_i(\bxi(\widehat{\bp'})) \geq u_0$ and hence $\bx(\widehat{\bp'})$ is an optimal allocation when $\epsilon = 0$.
\end{proof}
\begin{lemma}
\label{lem:sw-of-comp-prices}
Suppose that $\sum_j \qj = \sum_i \ei$. Then
\[
\sum_i \ei \log \ui(\exi(\bq)) - \sum_i \ei \log \ui(\exi(\bp'(l,\bq)))
\ge \sum_i \left( \frac 1l - 1 \right) (l \pjs - \qj) \cdot \mathsf{1}_{\qj \le l \pjs}.
\]
\end{lemma}
\begin{proof}
By Lemma~\ref{lem::reserve::1}, There exists an $\bx(\mathbf{p}')$, where $\mathbf{p'} = p'_j(l, \mathbf{q})$, such that if $\frac{p'_j(l, \mathbf{q})}{p^*_j} = l  $, $\sum_i x_{ij}(\mathbf{p}') \geq \frac{1}{l}$ 
and if $\frac{p'_j(l, \mathbf{q})}{p^*_j} = t  $, $\sum_i x_{ij}(\mathbf{p}') \leq \frac{1}{t}$. Therefore, by lemma~\ref{lem:p-pprime-comp},
\begin{align*}
& \sum_i e_i \log u_i(\bxi(\mathbf{q})) - \sum_i e_i \log u_i(\bxi(\mathbf{p}'(l, \mathbf{q}))) \\
& \geq \sum_{ij} (p'_j(l, \mathbf{q}) - q_j) x_{ij}(\mathbf{p}'(l, \mathbf{q})) \\
& \geq \sum_j  (p'_j(l, \mathbf{q}) - q_j) \cdot \mathsf{1}_{\frac{q_j}{p^*_j} \leq l } \cdot \frac{1}{l} + \sum_j  ( p'_j(l, \mathbf{q}) - q_j) \cdot \mathsf{1}_{\frac{q_j}{p^*_j} \geq t } \cdot \frac{1}{t}.
\end{align*}

Since $\sum_j p'_j (l, \mathbf{q}) = \sum_i e_i = \sum_j \qj$, and as $\qj = p'_j(l, \bq)$ when $ l < \frac{q_j}{p^*_j} < t$,
\[\sum_j  (p'_j(l, \mathbf{q}) - q_j) \cdot \mathsf{1}_{\frac{q_j}{p^*_j} \leq l } = - \sum_j  ( p'_j(l, \mathbf{q}) - q_j) \cdot \mathsf{1}_{\frac{q_j}{p^*_j} \geq t }.\] 

Thus
\begin{align*}
\sum_i e_i \log u_i(\bxi(\mathbf{q})) - \sum_i e_i \log u_i(\bxi(\mathbf{p}'(l, \mathbf{q})))
& \ge \sum_j (\frac{1}{l} - \frac{1}{t}) (p'_j(l, \mathbf{q}) - q_j) \cdot \mathsf{1}_{\frac{q_j}{p^*_j} \leq l } \\
& \ge \sum_j (\frac{1}{l} - 1) (l p^*_j - q_j) \cdot \mathsf{1}_{q_j \leq l p^*_j}.  
\end{align*}
\end{proof}

\begin{lemma}
\label{lem:pstar-comp}
$
 \sum_i \ei \log \ui(\exi(\bp'(l,\bq)))  \ge \sum_i \ei \log \ui(\exi(\bp^*).
$
\end{lemma}
\begin{proof}
The result follows from~\eqref{multi-poa-ineq}, as $\sum_j \pjs = \sum_i \ei = \sum_j \pj'$.
\end{proof}

The proof of the next Corollary uses Lemma~\ref{lem::price}, which depends on Lemma~\ref{lem::prices::les}, whose proof, for single-demand utility functions, changes slightly in the presence of reserve prices, as we comment on next.
\newline

\textsc{Revised Proof of Lemma~\ref{lem::prices::les}.}
There are two changes.
First, when $l \in S$, we can conclude that $p_l = p'_l > \max\{0, r_l\}$ (in the line preceding \eqref{eqn:p-pp-relation}).
Second, when $p_l > \max\{0, r_l\}$, $\sum_{h \neq i} x'_{hl} = \sum_{h \neq i} x_{hl} = 1$ (in the line preceding \eqref{eqn::fish::price}).
These are the only changes.

\begin{corollary}
\label{cor:second-bound}
\begin{align*}
& \sum_i \ei \log \ui(\vi, \bmi)
- \sum_i \ei \log u_i(\bxi(\bp^*))  \\
&~~~~ \ge 
 \sum_j (1 - l) \pjs - \left( \frac 1l - 1 \right) (\pj + \max_{i'} e_{i'}) \frac{ \sum_i \ei} {\sum_j (\pj + \max_{i'} e_{i'})} 
                                         \cdot \mathsf{1}_{(\pj + \max_{i'} e_{i'})  \frac{ \sum_i \ei} {\sum_j (\pj + \max_{i'} e_{i'})} \le l \pjs} \\
&~~~~~~~~  - \sum_i \ei \frac {\sum_j (\pj + \max_{i'} e_{i'})} {\sum_i \ei}.
\end{align*}
\end{corollary}
\begin{proof}
\begin{align*}
& \sum_i \ei \log \ui(\vi, \bmi) - \sum_i \ei \log u_i(\bxi(\bp^*))  \\
&~~~~\ge \sum_i \ei \log \ui(\exi(\bp + \mathsf{1} \cdot \max_{i'} e_{i'})) - \sum_i \ei \log u_i(\bxi(\bp^*))  \\
&~~~~\ge  \sum_i \ei \log \ui (\exi(\bq))
- \sum_i \ei \log \frac {\sum_j (\pj + \max_{i'} e_{i'})} {\sum_i \ei} \\
&~~~~~~~~~~~~~~- \sum_i \ei \log u_i(\bxi(\bp^*)) ~~~~\text{(by Lemma~\ref{lem:log-relationship})} \\
&~~~~\ge 
\sum_i \ei \log \ui(\exi(\bp'(l,\bq)))
+ \sum_i \left( \frac 1l - 1 \right) (l \pjs - \qj) \cdot \mathsf{1}_{\qj \le l \pjs}\\
&~~~~~~~~ - \sum_i \ei \log \frac {\sum_j (\pj + \max_{i'} e_{i'})} {\sum_i \ei} - \sum_i \ei \log u_i(\bxi(\bp^*)) ~~~~\text{(by Lemma~\ref{lem:sw-of-comp-prices})}\\
& ~~~~\ge  \sum_i \left( \frac 1l - 1 \right) (l \pjs - \qj) \cdot \mathsf{1}_{\qj \le l \pjs}  - \sum_i \ei \log \frac {\sum_j (\pj + \max_{i'} e_{i'})} {\sum_i \ei} ~~~~\text{(by Lemma~\ref{lem:pstar-comp})}.
\end{align*}
\end{proof}

Let $\del = \sum_j (\pj +\max_{i'} e_{i'}) - \sum_i e_i$.
Suppose that the good $j$ reserve price $\rj \le \frac 14 \pjs$ for all $j$.
Note that $\sum_i \ei + \sum_j \rj \cdot \mathsf{1}_{\pj = \rj} \ge \sum_j \pj \ge \sum_i \ei$.
Clearly, $\sum_j \rj \cdot \mathsf{1}_{\pj = \rj}  \ge \del - m \cdot \max_{i'} e_{i'}$.

\begin{lemma}
\label{lem:bound-extra-term}
\[
\sum_i \ei \log \frac {\sum_j (\pj + \max_{i'} e_{i'})} { \sum_i \ei} \le \del.
\]
\end{lemma}
\begin{proof}
\[
\sum_i e_i \log  \frac{\sum_j (p_j + \max_{i'} e_{i'})}{\sum_i e_i } = \sum_i e_i \log (1 + \frac{\delta}{\sum_i e_i}) \leq \sum_i e_i \frac{\delta}{\sum_i e_i}  = \delta.
\]
\end{proof}

\pfof{Lemma~\ref{lem:final-bound}}
We set $l = \frac 12$. Then by Corollary~\ref{cor:second-bound}
and Lemma~\ref{lem:bound-extra-term},
\begin{align*}
&\sum_i \ei \log \ui(\vi, \bmi)) - \sum_i \ei \log u_i(\bxi(\bp^*))  \\
& ~~~~ \geq \sum_j (2 - 1) (\frac{1}{2} p^*_j - (p_j + \max_{i'} e_{i'}) \cdot \frac{\sum_i e_i }{\sum_j (p_j + \max_{i'} e_{i'})} ) \\
&~~~~~~~~~~~~~~~~~~~\cdot\mathsf{1}_{(p_j + \max_{i'} e_{i'}) \cdot \frac{\sum_i e_i }{\sum_j (p_j + \max_{i'} e_{i'})} \geq \frac{1}{2} p^*_j}  - \delta \\
& ~~~~ \geq \sum_j (2 - 1) (\frac{1}{2} p^*_j - (r_j + \max_{i'} e_{i'}) \cdot \frac{\sum_i e_i }{\sum_j (p_j + \max_{i'} e_{i'})} ) \\
&~~~~~~~~~~~~~~~~~~~\cdot \mathsf{1}_{p_j = r_j \wedge (r_j + \max_{i'} e_{i'}) \cdot \frac{\sum_i e_i }{\sum_j (p_j + \max_{i'} e_{i'})} \leq \frac{1}{2} p^*_j} - \delta \\
&~~~~ \geq  \sum_j  (\frac{1}{2} p^*_j - (r_j + \max_{i'} e_{i'}) \cdot \frac{\sum_i e_i }{\sum_j (p_j + \max_{i'} e_{i'})} ) \cdot \mathsf{1}_{p_j = r_j \wedge (r_j + \max_{i'} e_{i'}) \leq \frac{1}{2} p^*_j} - \delta\\
&~~~~~~~~~~~~~~~~~~~~~~~~~~\mbox{(as $\frac{\sum_i e_i }{\sum_j (p_j + \max_{i'} e_{i'})} \leq 1$)} \\
&~~~~ \geq  \sum_j (\frac{1}{2} p^*_j - (r_j + \max_{i'} e_{i'})) \cdot \mathsf{1}_{p_j = r_j \wedge (r_j + \max_{i'} e_{i'}) \leq \frac{1}{2} p^*_j} - \delta~~~~~~\mbox{(as $\frac{\sum_i e_i }{\sum_j (p_j + \max_{i'} e_{i'})} \leq 1$)} \\
&~~~~ \geq  \sum_j (\frac{1}{2} p^*_j - (r_j + \max_{i'} e_{i'})) \cdot \mathsf{1}_{p_j = r_j \wedge r_j  \leq \frac{1}{2} p^*_j} - \delta \\
& ~~~~ \geq \sum_j r_j \cdot \mathsf{1}_{p_j = r_j}  - m  \cdot \max_{i'} e_{i'}  - \delta ~~~~~~~(\text{as}~\frac{1}{2} p^*_j \geq 2 r_j) \\
&~~~~ \ge \del -m\cdot \max_{i'} e_{i'} -  m \cdot \max_{i'} e_{i'}- \del = -2 m \cdot \max_{i'} e_{i'}~~~~(\text{as}~\sum_j r_j \cdot \mathsf{1}_{p_j = r_j} \ge \del - m  \cdot \max_{i'} e_{i'}).
\end{align*}
\end{proof}

\end{document}